\documentclass[11pt]{article}

\usepackage[utf8]{inputenc}
\usepackage[T1]{fontenc}
\usepackage{lmodern}
\usepackage[DIV=11]{typearea} % large value = less margin (roughly speaking), recommended values: 10pt: 8, 11pt: 10, 12pt: 12

\usepackage{amsmath}
\usepackage{amssymb}
\usepackage{amsthm}
\usepackage{thmtools}

\usepackage[ruled,vlined,linesnumbered]{algorithm2e}

\usepackage{tkz-graph}

\usepackage{xspace}
\usepackage{enumerate}
\usepackage{array}
\usepackage{subfigure}

\usepackage[
	style=alphabetic,
	backref=true,
	doi=false,
	url=false,
	maxcitenames=3,
	mincitenames=3,
	maxbibnames=10,
	minbibnames=10,
	backend=bibtex8,
	sortlocale=en_US
]{biblatex}

\usepackage[ocgcolorlinks]{hyperref} % Option ocgcolorlinks makes links only colored on screen and not on printouts

\newcommand*\samethanks[1][\value{footnote}]{\footnotemark[#1]}

\makeatletter
\g@addto@macro\bfseries{\boldmath}
\makeatother

\makeatletter
\g@addto@macro\mdseries{\unboldmath}
\g@addto@macro\normalfont{\unboldmath}
\g@addto@macro\rmfamily{\unboldmath}
\g@addto@macro\upshape{\unboldmath}
\makeatother

\DeclareCiteCommand{\citem}
    {}
    {\mkbibbrackets{\bibhyperref{\usebibmacro{postnote}}}}
    {\multicitedelim}
    {}

\renewcommand*{\multicitedelim}{\addcomma\space}

\AtEveryCitekey{\clearfield{note}}

\newcommand{\myhref}[1]{%
  \iffieldundef{doi}
    {\iffieldundef{url}
       {#1}
       {\href{\strfield{url}}{#1}}}
    {\href{http://dx.doi.org/\strfield{doi}}{#1}}%
}

\DeclareFieldFormat{title}{\myhref{\mkbibemph{#1}}}
\DeclareFieldFormat
  [article,inbook,incollection,inproceedings,patent,thesis,unpublished]
  {title}{\myhref{\mkbibquote{#1\isdot}}}

\makeatletter
\AtBeginDocument{%
    \newlength{\temp@x}%
    \newlength{\temp@y}%
    \newlength{\temp@w}%
    \newlength{\temp@h}%
    \def\my@coords#1#2#3#4{%
      \setlength{\temp@x}{#1}%
      \setlength{\temp@y}{#2}%
      \setlength{\temp@w}{#3}%
      \setlength{\temp@h}{#4}%
      \adjustlengths{}%
      \my@pdfliteral{\strip@pt\temp@x\space\strip@pt\temp@y\space\strip@pt\temp@w\space\strip@pt\temp@h\space re}}%
    \ifpdf
      \typeout{In PDF mode}%
      \def\my@pdfliteral#1{\pdfliteral page{#1}}% I don't know why % this command...
      \def\adjustlengths{}%
    \fi
    \ifxetex
      \def\my@pdfliteral #1{}% isn't equivalent to this one
      \def\adjustlengths{\setlength{\temp@h}{-\temp@h}\addtolength{\temp@y}{1in}\addtolength{\temp@x}{-1in}}%
    \fi%
    \def\Hy@colorlink#1{%
      \begingroup
        \ifHy@ocgcolorlinks
          \def\Hy@ocgcolor{#1}%
          \my@pdfliteral{q}%
          \my@pdfliteral{7 Tr}% Set text mode to clipping-only
        \else
          \HyColor@UseColor#1%
        \fi
    }%
    \def\Hy@endcolorlink{%
      \ifHy@ocgcolorlinks%
        \my@pdfliteral{/OC/OCPrint BDC}%
        \my@coords{0pt}{0pt}{\pdfpagewidth}{\pdfpageheight}%
        \my@pdfliteral{F}% Fill clipping path (the url's text) with
        \my@pdfliteral{EMC/OC/OCView BDC}%
        \begingroup%
          \expandafter\HyColor@UseColor\Hy@ocgcolor%
          \my@coords{0pt}{0pt}{\pdfpagewidth}{\pdfpageheight}%
          \my@pdfliteral{F}% Fill clipping path (the url's text)
        \endgroup%
        \my@pdfliteral{EMC}%
        \my@pdfliteral{0 Tr}% Reset text to normal mode
        \my@pdfliteral{Q}%
      \fi
      \endgroup
    }%
}
\makeatother

\colorlet{DarkRed}{red!50!black}
\colorlet{DarkGreen}{green!50!black}
\colorlet{DarkBlue}{blue!50!black}

\hypersetup{
	linkcolor = DarkRed,
	citecolor = DarkGreen,
	urlcolor = DarkBlue,
	bookmarks = true,
	bookmarksnumbered = true,
	linktocpage = true
}

\allowdisplaybreaks[1]

\usetikzlibrary{shapes}

\declaretheorem[numberwithin=section]{theorem}
\declaretheorem[numberlike=theorem]{lemma}
\declaretheorem[numberlike=theorem]{proposition}
\declaretheorem[numberlike=theorem]{corollary}
\declaretheorem[numberlike=theorem]{definition}
\declaretheorem[numberlike=theorem]{claim}

\newcounter{theoremtmp}
\newcounter{ValueIterationCounter}
\newcounter{TheoremThree}
\newcounter{TheoremTwo}
\newcounter{TheoremOne}

\DontPrintSemicolon
\SetKw{KwAnd}{and}
\SetProcNameSty{textsc}
\SetFuncSty{textsc}

\newcommand{\poly}{\operatorname{poly}}
\newcommand{\en}{\ensuremath{e}\xspace}
\newcommand{\counter}{\operatorname{count}}
\newcommand{\degree}{\operatorname{deg}}
\newcommand{\val}{\operatorname{val}}

\def\cA{\mathcal{A}}

\def\Alice{Alice\xspace}
\def\Bob{Bob\xspace}

\title{Polynomial-Time Algorithms for Energy Games with Special Weight Structures\thanks{This paper appeared in the ESA 2012 special issue of \emph{Algorithmica}~\cite{ChatterjeeHKN14}. A preliminary version was presented at the \emph{20th Annual European Symposium on Algorithms (ESA 2012).}}}

\author{
Krishnendu~Chatterjee\thanks{Institute of Science and Technology, Klosterneuburg, Austria. Supported by the Austrian Science Fund (FWF): P23499-N23, the Austrian Science Fund (FWF): S11407-N23 (RiSE), an ERC Start Grant (279307: Graph Games), and a Microsoft Faculty Fellows Award.}
\and
Monika~Henzinger\thanks{University of Vienna, Faculty of Computer Science, Vienna, Austria. Supported by the Austrian Science Fund (FWF): P23499-N23, the Vienna Science and Technology Fund (WWTF) grant ICT10-002, the University of Vienna (IK \mbox{I049-N}), and a Google Faculty Research Award.}
\and
Sebastian~Krinninger\samethanks[3]
\and
Danupon~Nanongkai\thanks{Nanyang Technological University, Singapore, Singapore. Work partially done while at University of Vienna, Austria.}
}

\date{}

\hypersetup{
	pdftitle = {Polynomial-Time Algorithms for Energy Games with Special Weight Structures},
	pdfauthor = {Krishnendu Chatterjee, Monika Henzinger, Sebastian Krinninger, Danupon Nanongkai}
}

\begin{document}
\maketitle
\begin{abstract}
Energy games belong to a class of {\em turn-based two-player infinite-duration games} played on a weighted directed graph. It is one of the rare and intriguing combinatorial problems that lie in
${\sf NP}\cap {\sf co\mbox{-}NP}$, but are not known to be in {\sf P}. The existence of polynomial-time algorithms has been a major open problem for decades and apart from pseudopolynomial algorithms there is no algorithm that solves any non-trivial subclass in polynomial time.

In this paper, we give several results based on the weight structures of the graph. First, we identify a notion of {\em penalty} and present a polynomial-time algorithm when the penalty is large. Our algorithm is the first polynomial-time algorithm on a large class of weighted graphs. It includes several worst-case instances on which previous algorithms, such as value iteration and random facet algorithms, require at least sub-exponential time.
Our main technique is developing the first non-trivial {\em approximation} algorithm and showing how to convert it to an exact algorithm.
Moreover, we show that in a practical case in verification where weights are clustered around a constant number of values, the energy game problem can be solved in polynomial time.
We also show that the problem is still as hard as in general when the clique-width is bounded or the graph is strongly ergodic, suggesting that restricting the graph structure does not necessarily help.

\end{abstract}
\newpage

\tableofcontents
\newpage

\section{Introduction}\label{sec:intro}

Consider a coffee shop $A$ having a budget of $\en$ competing with its rival $B$ across the street who has an unlimited budget. Each competitor can set the price of a cup of coffee between $ 1 $ cent and $ 10 $ euros (as an integer cent amount). Coffee shop $B$ can observe the price of a cup of coffee at $A$, say $p_0$, and responds with a price $p_1$, causing $A$ a loss of $w(p_0, p_1)$, which could potentially put $A$ out of business. If $A$ manages to survive, then it can respond to $B$ with a price $p_2$, gaining itself a profit of $w(p_1, p_2)$. Then $B$ will try to put $A$ out of business again with a price $p_3$. How much initial budget $\en$ does $A$ need in order to guarantee that its business will survive forever?
This is a simple example of a perfect-information turn-based infinite-duration game called an {\em energy game}, defined as follows.

In an energy game, there are two players, \Alice and \Bob, playing a game on a finite directed graph $G=(V, E)$ with weight function $w: E\rightarrow \mathbb{Z}$. Each node in $G$ belongs to either \Alice or \Bob.
The game starts by placing an imaginary car on a specified starting node $v_0$ with an initial energy $\en_0\in \mathbb{Z}^{\geq 0} \cup \{\infty\}$ in the car (where $\mathbb{Z}^{\geq 0} = \{0, 1, \ldots\}$).
The game is played in {\em rounds}:
at any round $i>0$, if the car is at node $v_{i-1}$ and has energy $\en_{i-1}$, then the owner of $v_{i-1}$ moves the car from $v_{i-1}$ to a node $v_i$ along an edge $(v_{i-1}, v_i)\in E$. The energy of the car is then updated to $\en_i=\en_{i-1}+w(v_{i-1}, v_i)$.
The goal of \Alice is to sustain the energy of the car while \Bob will try to make \Alice fail. That is, we say that
\Alice \emph{wins} the game if the energy of the car is never below zero, i.e. $\en_i\geq 0$ for all~$i$; otherwise, \Bob wins.
The problem of {\em computing the minimal sufficient energy} is to compute the minimal initial energy $\en_0$ such that \Alice wins the game. (Note that such $\en_0$ always exists since it could be $\infty$ in the worst case.)
Figure~\ref{fig:energy_games_example} shows an example run of an energy game.
The important parameters in terms of running time are the number $ n $ of nodes in the graph, the number $ m $ of edges in the graph, and the weight parameter $ W $ defined as $ W = \max_{(u, v) \in E} | w(u, v) | $.

\begin{figure}[h]
\centering

\subfigure[][]{
\label{fig:energy_games_example_a}
\begin{tikzpicture}
\SetGraphUnit{1.5}
\tikzset{VertexStyle/.style = {minimum size=15pt,inner sep=0pt,outer sep=0pt,draw}}
\tikzset{EdgeStyle/.style = {->}}
\tikzset{LabelStyle/.style= {draw=none,inner sep=3pt,outer sep=0pt,fill=lightgray!50}}

\Vertex[NoLabel,style={shape=circle}]{a}
\NOWE[NoLabel,style={shape=rectangle}](a){b}
\NOEA[NoLabel,style={shape=rectangle}](a){c}

\Edge[label=7](a)(b)
\Edge[label=2, style={ultra thick}](a)(c)
\Edge[label=4](b)(c)
\Edge[label=-2, style={->,bend right=60}](b)(a)
\Edge[label=3, style={->,bend right=60}](c)(b)
\Edge[label=-8, style={->,bend left=60}](c)(a)

\node at (0, 0) {\includegraphics[width=10pt]{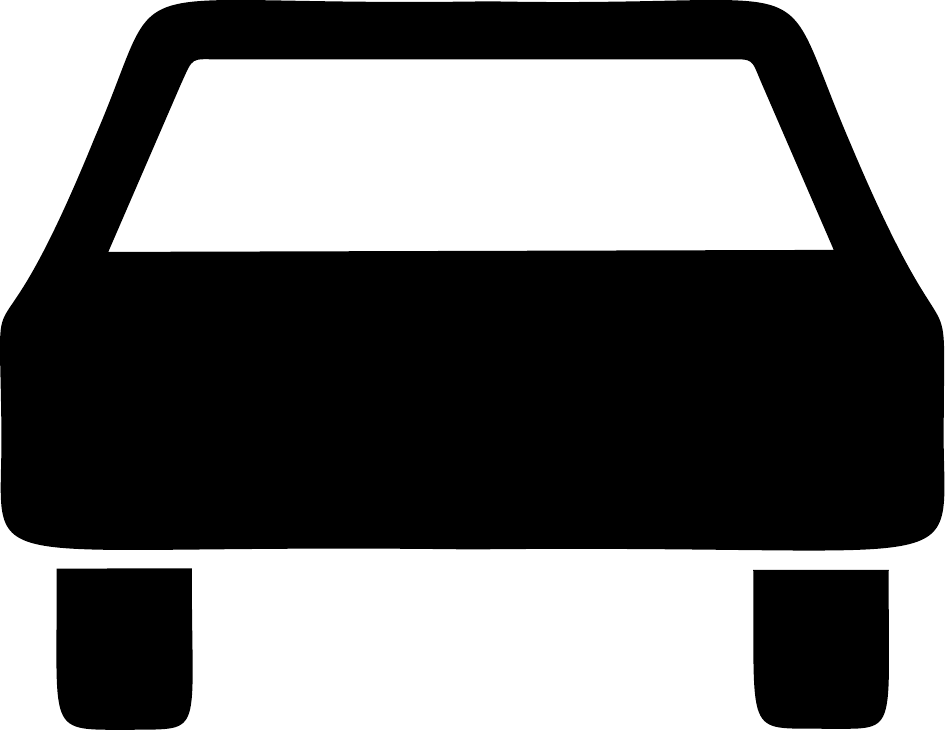}};
\node at (0, -0.75) {
\begin{tabular}{m{25pt} m{1em}}
\includegraphics[width=14pt,angle=90]{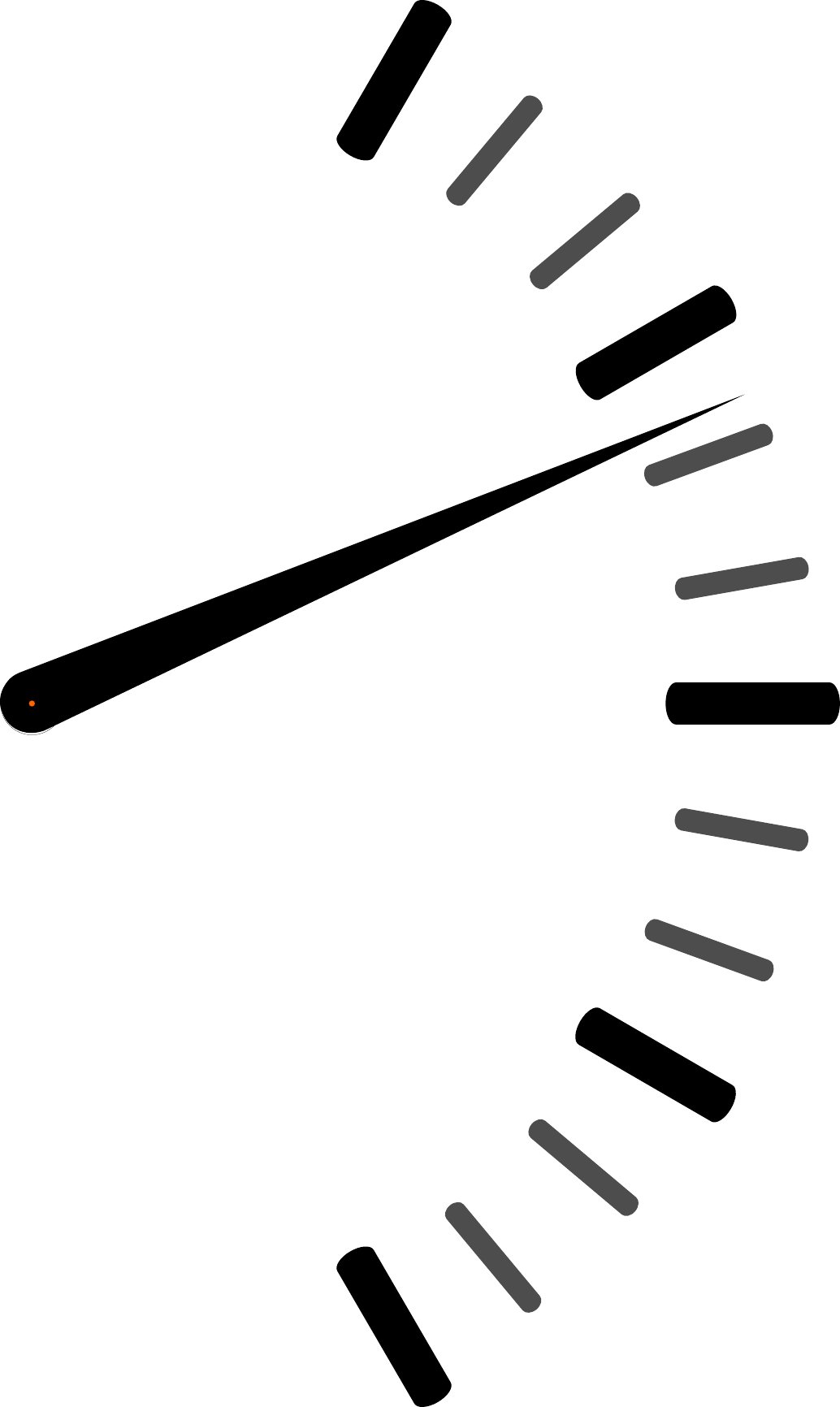} &
\begin{tikzpicture}
\node[draw,fill=lightgray!50] {0};
\end{tikzpicture}
\end{tabular}
};
\end{tikzpicture}
}
\hspace{0.5em}
%%%%%%%%%%%%%%%%%%%%%%%%%%%%%%%%%%%%%%%%%%%%%%%%%%%%%%%%%%%%%%%%%%%%%%%%%%%%%%%%%%%%%%%%
\subfigure[][]{
\label{fig:energy_games_example_b}
\begin{tikzpicture}
\SetGraphUnit{1.5}
\tikzset{VertexStyle/.style = {minimum size=15pt,inner sep=0pt,outer sep=0pt,draw}}
\tikzset{EdgeStyle/.style = {->}}
\tikzset{LabelStyle/.style= {draw=none,inner sep=3pt,outer sep=0pt,fill=lightgray!50}}

\Vertex[NoLabel,style={shape=circle}]{a}
\NOWE[NoLabel,style={shape=rectangle}](a){b}
\NOEA[NoLabel,style={shape=rectangle}](a){c}

\Edge[label=7](a)(b)
\Edge[label=2](a)(c)
\Edge[label=4](b)(c)
\Edge[label=-2, style={->,bend right=60}](b)(a)
\Edge[label=3, style={->,bend right=60}](c)(b)
\Edge[label=-8, style={style={ultra thick},->,bend left=60}](c)(a)

\node at (1.5, 1.5) {\includegraphics[width=10pt]{fig/car}};
\node at (0, -0.75) {
\begin{tabular}{m{25pt} m{1em}}
\includegraphics[width=14pt,angle=90]{fig/fuel_meter} &
\begin{tikzpicture}
\node[draw,fill=lightgray!50] {2};
\end{tikzpicture}
\end{tabular}
};
\end{tikzpicture}
}
\hspace{0.5em}
%%%%%%%%%%%%%%%%%%%%%%%%%%%%%%%%%%%%%%%%%%%%%%%%%%%%%%%%%%%%%%%%%%%%%%%%%%%%%%%%%%%%%%%%
\subfigure[][]{
\label{fig:energy_games_example_c}
\begin{tikzpicture}
\SetGraphUnit{1.5}
\tikzset{VertexStyle/.style = {minimum size=15pt,inner sep=0pt,outer sep=0pt,draw}}
\tikzset{EdgeStyle/.style = {->}}
\tikzset{LabelStyle/.style= {draw=none,inner sep=3pt,outer sep=0pt,fill=lightgray!50}}

\Vertex[NoLabel,style={shape=circle}]{a}
\NOWE[NoLabel,style={shape=rectangle}](a){b}
\NOEA[NoLabel,style={shape=rectangle}](a){c}

\Edge[label=7](a)(b)
\Edge[label=2](a)(c)
\Edge[label=4](b)(c)
\Edge[label=-2, style={->,bend right=60}](b)(a)
\Edge[label=3, style={->,bend right=60}](c)(b)
\Edge[label=-8, style={->,bend left=60}](c)(a)

\node at (0, 0) {\includegraphics[width=10pt]{fig/car}};
\node at (0, -0.75) {
\begin{tabular}{m{25pt} m{1em}}
\includegraphics[width=14pt,angle=90]{fig/fuel_meter} &
\begin{tikzpicture}
\node[draw,fill=lightgray!50] {-6};
\end{tikzpicture}
\end{tabular}
};
\end{tikzpicture}
}

\caption{
An example of an energy game.
The round node belongs to Alice and the rectangular nodes belong to Bob.
The current energy level of the car is written in the box at the bottom.
The game starts at the bottom node, which belongs to Alice, with the initial energy level $ 0 $~\subref{fig:energy_games_example_a}.
Alice chooses to move the car to the upper right node using the edge of weight $ 2 $.
Afterwards the car has energy $ 2 $ and is located on Bob's node~\subref{fig:energy_games_example_b}.
Bob chooses to move the car to the bottom node using the edge of weight $ -8 $.
This decreases the energy of the car to $ -6 $~\subref{fig:energy_games_example_c}.
At this point Alice has lost the game because the energy of the car is negative.
}\label{fig:energy_games_example}
\end{figure}

\paragraph{Related Work.} Energy games belong to an intriguing family of {\em infinite-duration turn-based games} which includes alternating games \cite{RothBKM10}, and has applications in areas such as
computer-aided verification and automata theory~\cite{CdAHS03,BCHJ09,CCHRS11}, as well as in online and streaming problems~\cite{ZwickP96}.
Energy games are polynomial-time equivalent to \emph{mean-payoff games}~\cite{Bouyer08}.
Furthermore there are polynomial-time reductions from \emph{parity games} to energy games~\cite{Jurdzinski98} and from energy games to \emph{simple stochastic games}~\cite{Condon92,ZwickP96}.\footnote{Both reductions were originally shown for mean-payoff games.}
These games are among the rare combinatorial problems, along with {\em Graph Isomorphism}, that are unlikely to be {\sf NP}-complete (since they are in ${\sf UP} \cap {\sf co\mbox{-}UP} \subseteq {\sf NP}\cap {\sf co\mbox{-}NP}$~\cite{EM79,GurvichKK88,ZwickP96,Jurdzinski98}) but not known to be in {\sf P}. 
It is a major open problem whether any of these games are in {\sf P} or not.
While the energy game is relatively new and interesting in its own right, it has been implicitly studied since the late 80s, due to its close connection with the mean-payoff game. In particular, the seminal paper by Gurvich et al.~\cite{GurvichKK88} presents a simplex-like algorithm for mean-payoff games which computes a ``potential function'' that is essentially the energy function.\footnote{More precisely, the auxiliary algorithm of Gurvich et al.~\cite{GurvichKK88} solves a decision version of mean-payoff games. It has to output, for every node $ v $, whether the mean-payoff at $ v $ is at least zero or not. If it is, the potential of $ v $ computed by this algorithm is equal to what we call the minimal sufficient energy of $ v $. If not, we know that the minimal sufficient energy of $ v $ is $ \infty $.}

The algorithm of Gurvich et al.~\cite{GurvichKK88} was shown to be pseudopolynomial by Pisaruk~\cite{Pisaruk99}.
Another pseudopolynomial algorithm was given by Zwick and Paterson~\cite{ZwickP96}.
Björklund and Vorobyov~\cite{Bjoerklund07} developed an algorithm for mean-payoff games that besides being pseudopolynomial has a randomized strongly subexponential running time of $2^{O(\sqrt{n\log n})} \log W$.
Lifshits and Pavlov~\cite{Lifshits07} described an exponential algorithm for mean-payoff games.
Recently, Brim et al.~\cite{Brim11} gave an algorithm for energy games that is faster than previous deterministic pseudopolynomial algorithms and runs in time $ O (m n W) $.
It yields the current fastest pseudopolynomial complexity for energy games as well as mean-payoff games.
To the best of our knowledge only two special cases of energy or mean-payoff games are known to admit a polynomial-time algorithm.
The first special case is where all nodes belong to one player and for example can be solved with Karp's minimum cycle mean algorithm~\cite{Karp1978}.
The second special case is when $ W $ is polynomial in the input size and thus pseudopolynomial algorithms actually run in polynomial time.

Infinite-duration turn-based games also have strong connections to {\em (mixed) Nash equilibrium computation}~\cite{DaskalakisP11} and {\em Linear Programming}~\cite{Vorobyov08}.
For example, they are in a low complexity class lying very close to {\sf P} called {\sf CCLS} \cite{DaskalakisP11} which is in ${\sf PPAD}\cap {\sf PLS}$.
This implies that, unlike many problems in Game Theory, these games are {\em unlikely} to be {\sf PPAD}-complete.
Moreover, as shown by Halman~\cite{Halman07algorithmica}, all these games are \emph{LP-type problems}~\cite{SharirW92}, a concept that generalizes linear programming.
Therefore the random facet algorithm~\cite{Kalai92,Kalai97,MatousekSW96}, a simplex-algorithm with a certain randomized pivoting rule, can be used to solve them in randomized subexponential time\footnote{The randomized subexponential algorithm of Björklund and Vorobyov~\cite{Bjoerklund07} uses the same randomization scheme as the random facet algorithm.}.
This relates infinite-duration turn-based games to the question whether there exists a pivoting rule for the simplex algorithm that requires a polynomial number of pivoting steps on any linear program, which is perhaps one of the most important problems in the field of linear programming.
In fact, several randomized pivoting rules have been conjectured to solve linear programs in polynomial time until recent breakthrough results~\cite{FriedmannHZSTOC11,FriedmannIPCO11,FriedmannHZSODA11} have rejected these conjectures.
As noted by Friedmann et al.~\cite{FriedmannHZSTOC11}, infinite-duration turn-based games played an important role in this breakthrough as the lower bounds were first developed for these games and later extended to linear programs via Markov Decision Processes.
Figure~\ref{fig:energy_games_complexity} summarizes the complexity status of energy games and related problems.

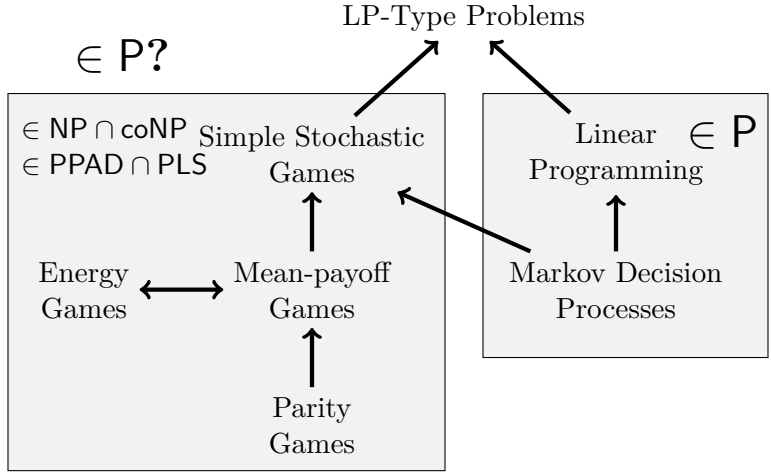
\begin{figure}
\centering
\begin{tikzpicture}
\filldraw[fill=gray!10] (-6,-6) rectangle (-0.25, -1);
\filldraw[fill=gray!10] (0.25,-1) rectangle (4,-4.5);

\node (LPtype) at (0,0) [draw=none,align=center] {LP-Type Problems};
\node (SSG) at (-2,-1.8) [draw=none,align=center] {Simple Stochastic \\ Games};
\node (LP)  at (2,-1.8) [draw=none,align=center]{Linear \\ Programming};
\node (EG) at (-5,-3.6) [draw=none,align=center] {Energy \\Games};
\node (MPG) at (-2,-3.6) [draw=none,align=center] {Mean-payoff \\ Games};
\node (MDP) at (2,-3.6) [draw=none,align=center] {Markov Decision \\ Processes};
\node (PG) at (-2,-5.4) [draw=none,align=center] {Parity \\ Games};

\draw [->,ultra thick] (SSG) to (LPtype);
\draw [->,ultra thick] (LP) to (LPtype);
\draw [->,ultra thick] (MPG) to (SSG);
\draw [->,ultra thick] (MDP) to (LP);
\draw [->,ultra thick] (MDP) to (SSG);
\draw [->,ultra thick] (PG) to (MPG);
\draw [<->,ultra thick] (EG) to (MPG);

\node at (3.4, -1.5) {\LARGE $ \in \mathsf{P} $};

\node at (-4.5, -0.5) {\LARGE $ \in \mathsf{P} $\textbf{?}};

\node at (-4.6, -1.7) [align=left] {$ \in \mathsf{NP} \cap \mathsf{coNP} $ \\ $ \in \mathsf{PPAD} \cap \mathsf{PLS} $ };

\end{tikzpicture}

\caption{The complexity status of energy games and related problems. Arrows indicate polynomial-time reductions.}\label{fig:energy_games_complexity}
\end{figure}

\paragraph{Our Contributions.} In this paper we identify several classes of graphs (based on weight structures) for which energy games can be solved in polynomial time.
For any starting node $s$, let $\en_{G, w}^*(s)$ denote the minimal sufficient energy.
Our first contribution is an algorithm whose running time is based on a parameter called {\em penalty}.
Informally, a penalty\footnote{We formally define the concept of penalty in Section~\ref{sec:prelim}.} of $ D $ means that Bob has a way to play optimally such that, for all choices of Alice, one of the following two situations occurs.
(1) Alice wins the game for some finite initial energy.
(2) Alice loses the game even if an additional energy of $ D $ would be added to the car in every turn.
We denote the \emph{penalty of the graph $ (G, w) $} by $ P(G, w) $.
We show that the \emph{higher} the penalty is, the faster we can compute the minimal energies.

\setcounter{TheoremOne}{\value{theorem}}
\begin{theorem}\label{thm:main 1}
Given a graph $(G, w)$ and an integer $M$ we can compute the minimal initial energies of all nodes in
\begin{equation*}
O \left( mn \left( \log\frac{M}{n} \right) \left( \log\frac{M}{n\lceil P(G, w)\rceil} \right) + m \frac{M}{\lceil P(G, w)\rceil} \right)
\end{equation*}
time, provided that for all $v$, $e_{G, w}^*(v)<\infty$ implies that $e_{G, w}^*(v)\leq M$.
\end{theorem}
We note that in addition to $(G, w)$, our algorithm takes $M$ as an input. If $M$ is unknown, we can simply use the universal upper bound $M=nW$~\cite{Brim11}. Allowing different values of $M$ will be useful in our proofs.
We emphasize that the algorithm can run without knowing $P(G, w)$. Our algorithm is as efficient as the fastest known pseudopolynomial-time ($O(mnW)$-time) algorithm \cite{Brim11} in the general case where $M=nW$ and 
$P(G, w)=1/n$ (so $\lceil P(G, w)\rceil =1$).

If the penalty is at least $W/\poly(n)$, our algorithm runs in polynomial time. Therefore, the algorithm also solves several classes of graphs that are previously not known to be solvable in polynomial time.
As an illustration, consider the class of graphs where each cycle has total weight either positive or less than $-W/2$. In this case, our algorithm runs in polynomial time.
{\em All known worst-case instances}~\cite{GurvichKK88,Beffara2001,ZwickP96,FriedmannHZSODA11} of previous algorithms fall in this class of graphs.
In particular, we observe that, for this class of graphs, the following algorithms need at least subexponential time (while our algorithm runs in polynomial time): the algorithm by Gurvich et al.~\cite{GurvichKK88}, the algorithm by Brim et al.~\cite{Brim11}, the algorithm by Zwick and Paterson~\cite{ZwickP96} and the random facet algorithm by Matoušek et al.~\cite{MatousekSW96} (the latter two algorithms are used for the decision versions of mean-payoff and parity games, respectively).\footnote{A worst-case instance for the first two algorithms has been developed by Lebedev and is mentioned by Gurvich et al.~\cite{GurvichKK88} and shown by Beffara and Vorobyov~\cite{Beffara2001} (for the second algorithm, we exploit the fact that it is deterministic and there exists a bad ordering in which the nodes are processed). Worst-case instances for the third and the fourth algorithm have been given by Zwick and Paterson~\cite{ZwickP96} and Friedmann et al.~\cite{FriedmannHZSODA11}, respectively. We note that the instances shown by Beffara and Vorobyov~\cite{Beffara2001} and Friedmann et al.~\cite{FriedmannHZSODA11} contain one cycle of small negative weight. One can change the value of this cycle to $ -W $ to make these examples belong to the desired class of graphs without changing the worst-case behaviors of the mentioned algorithms.}

Our result might also be of a practical interest since it solves energy games faster when penalties are high while it runs with the same running time as previous pseudopolynomial-time algorithms \cite{Brim11} in the worst case.

Our second contribution is an algorithm that approximates the minimal energy within some {\em additive error} where the size of the error depends on the penalty.
This result is the main tool in proving Theorem~\ref{thm:main 1} where we show how to use the approximation algorithm to compute the minimal energy {\em exactly}.

\setcounter{TheoremTwo}{\value{theorem}}
\begin{theorem}\label{thm:rounding_procedure_gives_desired_approximation}
Given a graph $(G, w)$ with $ P(G, w) \geq 1 $, an integer $M$, and an integer $c$ such that $n\leq c \leq n P(G, w)$, we can compute an energy function $ \en $ such that
\begin{align*}
\en (v) \leq \en^*_{G, w} (v) \leq \en (v) + c
\end{align*}
for every node $v$ in $O (mnM/c) $ time, provided that for every node $v$, $e_{G, w}^*(v)<\infty$ implies that $e_{G, w}^*(v)\leq M$.
\end{theorem}
The main technique in proving Theorem~\ref{thm:rounding_procedure_gives_desired_approximation} is rounding weights appropriately. We note that a similar idea of approximation has been explored earlier in the case of mean-payoff games~\cite{BorosEFGMM11}. Roth et al. \cite{RothBKM10} show an \emph{additive} FPTAS for rational weights in $ [-1, 1] $. This implies an additive error of $\epsilon W$ for any $\epsilon>0$ in our setting. This does not help in general since the error depends on $W$. Boros et al.~\cite{BorosEFGMM11} later achieved a {\em multiplicative} error of $(1+\epsilon)$. This result holds, however, only when the edge weights are non-negative integers. In fact, it is shown that if one can approximate the mean-payoff within a small multiplicative error in the general case, then the exact mean-payoff can be found~\cite{Gentilini14}. Despite several results for mean-payoff games, there is currently {\em no} approximation algorithm for general energy games. 
Our algorithm is the first non-trivial approximation algorithm for the energy game.

Our third contribution is a variant of the {\em Value Iteration Algorithm} by Brim et al.~\cite{Brim11} which runs faster in many cases. The running time of the algorithm depends on a concept that we call {\em admissible list} (defined in Section~\ref{sec:value_iteration_algorithm}) which uses the weight structure. One consequence of this result is used to prove Theorem~\ref{thm:rounding_procedure_gives_desired_approximation}. The other consequence is an algorithm for what we call the {\em fixed-window} case.

\setcounter{TheoremThree}{\value{theorem}}
\begin{theorem}\label{thm:main 3}\label{thm:algorithm_windowed_game}
If there are $ d $ values $ w_1, \ldots, w_d $ and a window size $ \delta $ such that for every edge $ (u, v) \in G $ we have $ w(u, v) \in \{ w_i - \delta, \ldots, w_i + \delta \} $ for some $ 1 \leq i \leq d $, then the minimal energies can be computed in $ O (m \delta n^{d+1}) $ time.
\end{theorem}

The fixed-window case, besides its theoretical attractiveness, is also interesting from a practical point of view. Energy and mean-payoff games have many applications in the area of verification, mainly in the synthesis of reactive systems with resource constraints~\cite{BCHJ09} and performance aware program synthesis~\cite{CCHRS11}. In most applications related to synthesis, the resource consumption is through only a few common operations, and each operation depending on the current state of the system consumes a related amount of resources. In other words, in these applications there are $ d $ groups of weights (one for each operation) where in each group the weights differ by at most $\delta$ (i.e, $\delta$ denotes the small variation in resource consumption for an operation depending on the current state), and $d$ and $\delta$ are typically constant. Theorem~\ref{thm:main 3} implies a polynomial-time algorithm for this case.

We also show that the energy game problem is still as hard as the general case even when the clique-width is bounded or the graph is strongly ergodic (see Section~\ref{sec:discussion}). This suggests that restricting the graph structures might not help in solving the problem, which is in sharp contrast to the fact that parity games can be solved in polynomial time in these cases~\cite{Obdrzalek07,Lebedev05}.
\begin{theorem}
The energy game problem on arbitrary graphs is polynomial-time equivalent to the energy game problem on graphs that have bounded clique-width as well as to the energy game problem on graphs that are strongly ergodic.
\end{theorem}

\section{Preliminaries}\label{sec:definition}\label{sec:prelim}
Figure~\ref{table:notations} summarizes the notation introduced in this section.
\begin{figure}[h]
\centering
\fbox{
\begin{tabular}{l p{9cm}}
$ G = (V, E) $ & Directed graph with nodes $ V $ and edges $ E $ in which every node has out-degree $ \geq 1 $ \\
$ V_A $ ($V_B$) & Set of nodes controlled by Alice (Bob). $ V_A \cup V_B = V $ and $ V_A \cap V_B = \emptyset $ \\
$ n $ & Number of nodes in $ G $, i.e., $ n = |V| $ \\
$ m $ & Number of edges in $ G $, i.e., $ m = |E| $ \\
$ w(u, v) $ & Weight of edge $ (u, v) $ \\
$ w(P) $ & Total weight of a finite path $ P $, sum of all edge weights on $ P $ \\
$ W $ & Maximum absolute edge weight, $ W = \max_{(u, v) \in E} | w(u, v) | $ \\
$ \en^*_{G, w} (v) $ & Minimal energy at node $ v $ in weighted graph $ (G, w) $ \\
%$ \en^*_{G, w} $ & Minimal energy function of weighted graph $ (G, w) $ \\
$ P (G, w) $ & Penalty of weighted graph $ (G, w) $ \\
$ \sigma $ ($ \tau $) & A strategy of Alice (Bob), i.e., a function that maps every node $ u \in V_A $ ($ u \in V_B $) to a neighboring node $ v $ such that $ (u, v) \in E $ \\
$ (\sigma, \tau) $ & A pair of strategies where $ \sigma $ is a strategy of Alice and $ \tau $ is a strategy of Bob \\
$ \sigma^* $ ($\tau^*$) & An optimal strategy of Alice (Bob) \\
$ G(\sigma, \tau) $ & Restriction of $ G $ to pair of strategies $ (\sigma, \tau) $
\end{tabular}
}
\caption{Overview of notation defined in Section~\ref{sec:prelim}}\label{table:notations}
\end{figure}

\paragraph{Energy Games.} An energy game is played by two players, Alice and Bob. Its input instance consists of a finite weighted directed graph $ (G, w) $ where all nodes have out-degree at least one\footnote{$(G, w)$ is usually called a ``game graph'' in the literature. We will simply say ``graph''.}.
The set of nodes $ V $ is partitioned into $ V_A $ and $ V_B $, which belong to Alice and Bob respectively, and every edge $ (u, v) \in E $ has an integer weight $ w(u, v) \in \{ -W, \ldots, W \} $.
It can be assumed without loss of generality that there are no self-loops.\footnote{If some node $ v $ has a self-loop $ (v, v) $ of weight $ w(v, v) $, we can replace the self-loop as follows: we add an artificial node $ v' $ and two edges $ (v, v') $ and $ (v',v ) $. The edges $ (v, v') $ and $ (v',v ) $ both get the weight $ w(v, v) $. This does not change the energy values nor the average weight of any cycle.}
Additionally, we are given a node $ s $ and an initial energy $\en_0$.
To formally define energy games, we need the notion of {\em strategies}. While general strategies can depend on the history of the game, it has been shown that we can assume that if a player wins a game, a {\em positional strategy} suffices to win~\cite{CdAHS03,Bouyer08}.\footnote{Positional strategies are both \emph{pure} and \emph{memoryless}, i.e., they are deterministic and do not depend on the history of the game. The existence of optimal positional strategies in energy games follows immediately from the existence of optimal positional strategies in mean-payoff games~\cite{EM79} and the reduction of energy games to mean-payoff games~\cite{Bouyer08}.} Therefore we only consider positional strategies. A positional strategy $\sigma$ of Alice is a mapping from each node in $V_A$ to one of its out-neighbors, i.e., for any $u\in V_A$, $\sigma(u)=v$ for some $(u,v)\in E$. This means that Alice sends the car to $v$ every time it is at $u$. We define a positional strategy $\tau$ of Bob similarly.
We simply use ``strategy'' instead of ``positional strategy'' in the rest of the paper.

A \emph{pair of strategies} $ (\sigma, \tau) $ consists of a strategy $ \sigma $ of Alice and $ \tau $ of Bob.
For any pair of strategies $ (\sigma, \tau) $, we define $G(\sigma, \tau)$ to be the subgraph of $G$ having only edges corresponding to the strategies~$ \sigma $ and~$ \tau $; i.e.,
\[G(\sigma, \tau)=(V, E') ~~~\mbox{where}~~~ E'=\{(u, \sigma(u))\mid u\in V_A\}\cup \{(u, \tau(u))\mid u\in V_B\}.\]
In $ G (\sigma, \tau) $ every node has a unique out-edge.

\label{sec:definition_energy}

Now, consider an energy game played by Alice and Bob starting at node $s$ with initial energy $\en_0$ using strategies $\sigma$ and $\tau$, respectively. We use $G(\sigma, \tau)$ to determine who wins the game as follows. For any $i$, let $P_i$ be the (unique) directed path of length $ i $ in $G(\sigma, \tau)$ originating at $s$. Observe that $P_i$ is exactly the path that the car will be moved along for $i$ rounds, and the energy of the car after~$i$ rounds is $\en_i=\en_0+w(P_i)$ where $w(P_i)$ is the sum of the edge weights in $P_i$. We say that {\em Bob wins} the game if there exists $i$ such that $\en_0+w(P_i)<0$ and {\em Alice wins} otherwise. Equivalently, we can determine who wins as follows. Let $C$ be the (unique) cycle reachable by $s$ in $G(\sigma, \tau)$, and let $w(C)$ be the sum of the edge weights in $C$. If $w(C)<0$, then Bob wins; otherwise, Bob wins if and only if there exists a {\em simple} path $P_i$ of some length $ i $ such that $\en_0+w(P_i)<0$.

This leads to the following definition of the {\em minimal sufficient energy} at node~$s$ \emph{corresponding to strategies} $\sigma$ and $\tau$, denoted by $\en^*_{G(\sigma, \tau), w}(s)$: If $w(C)<0$, then $\en^*_{G(\sigma, \tau), w}(s)=\infty$; otherwise, $\en^*_{G(\sigma, \tau), w}(s)=\max\{0, -\min w(P_i)\}$ where the minimization is over all {\em simple} paths $P_i$ in $G(\sigma, \tau)$ originating at $s$.
We then define the {\em minimal sufficient energy} at node $s$ to be
\begin{equation}
\en^*_{G, w}(s) = \min_\sigma \max_\tau \en^*_{G(\sigma, \tau), w}(s)\label{eq:energy definition}
\end{equation}
where the minimization and the maximization are over all positional strategies $\sigma$ of Alice and $\tau$ of Bob, respectively. We note that it follows from Martin's determinacy theorem~\cite{Mar75} that $\min_\sigma \max_\tau \en^*_{G(\sigma, \tau), w}(s) = \max_\tau\min_\sigma \en^*_{G(\sigma, \tau), w}(s)$, and thus it does not matter which player picks the strategy first. We say that a strategy $\sigma^*$ of Alice is an {\em optimal strategy} if for any strategy $\tau$ of Bob, $\en^*_{G(\sigma^*, \tau), w}(s)\leq \en^*_{G, w}(s)$. Similarly, $\tau^*$ is an optimal strategy of Bob if for any strategy $\sigma$ of Alice, $\en^*_{G(\sigma, \tau^*), w}(s)\geq \en^*_{G, w}(s)$.

We call any $\en: V\rightarrow \mathbb{Z}^{\geq 0} \cup \{\infty\}$ an {\em energy function}. We call $\en_{G, w}^*$ in Eq.~\eqref{eq:energy definition} a {\em minimal sufficient energy function} or simply a {\em minimal energy function}.
By this definition the minimal energy function is unique.
If $\en(s)\geq \en^*_{G, w}(s)$ for all $s$, then we say that $\en$ is a {\em sufficient} energy function.
The goal of the energy game problem is to compute $e^*_{G, w}$.

We say that a natural number $ M $ is an {\em upper bound on the finite minimal energy} if for every node $ v $ either $ \en^*_{G, w} (v) = \infty $ or $ \en^*_{G, w} (v) \leq M $. This means that every finite minimal energy is bounded from above by $ M $.
A universal upper bound is $ M = nW $~\cite{Brim11}.

\paragraph{Penalty.} Let $(G, w)$ be a weighted graph. For any node $s$ and real $D\geq 0$, we say that $s$ has a {\em penalty of at least $D$} if there exists an optimal strategy $\tau^*$ of Bob such that for any strategy $\sigma$ of Alice, the following condition holds for the (unique) cycle $ C $ reachable by $s$ in $G(\sigma, \tau^*)$: if $ w(C) < 0 $, then the average weight on $ C $ is at most $ -D $, i.e. $ \sum_{(u, v) \in C} w (u, v) / |C| \leq -D $. See Figure~\ref{fig:example_penalty} for an example.

\begin{figure}
\centering

\begin{tikzpicture}
\SetGraphUnit{1.5}
\tikzset{VertexStyle/.style = {minimum size=15pt,inner sep=0pt,outer sep=0pt,draw}}
\tikzset{EdgeStyle/.style = {->}}
\tikzset{LabelStyle/.style= {draw=none,inner sep=3pt,outer sep=0pt,fill=lightgray!50}}

\Vertex[NoLabel,style={shape=circle}]{a}
\NOWE[NoLabel,style={shape=rectangle}](a){b}
\NOEA[NoLabel,style={shape=rectangle}](a){c}

\Edge[label=7](a)(b)
\Edge[label=2](a)(c)
\Edge[label=4](b)(c)
\Edge[label=-8, style={->,bend left=60}](c)(a)

\node at (0, 0) {\includegraphics[width=10pt]{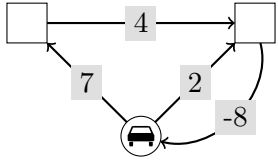}};
\end{tikzpicture}

\caption{
The graph of this picture is the modification of the graph in Figure~\ref{fig:energy_games_example} where we have fixed Bob's optimal strategy.
The round node belongs to Alice and the rectangular nodes belong to Bob.
Alice has two strategies at the bottom node.
If she chooses to go left, then the car runs into a cycle of total weight $ 3 $ and average weight $ 1 $.
If she goes right, the car runs into a cycle of total weight $ -6 $ and average weight $ -3 $.
Therefore the graph has penalty $ 3 $.
}\label{fig:example_penalty}
\end{figure}

Intuitively, this means that either Alice wins the game using a finite initial energy, or she loses {\em significantly}, i.e., even if she would constantly receive an extra energy of a little less than $D$ per round, she still needs an infinite initial energy in order to win the game. We note that $ \sum_{(u, v) \in C} w (u, v) / |C|$ is known in the literature as the {\em mean-payoff} of $s$ when Alice and Bob play according to $ \sigma $ and $ \tau^*$, respectively. Thus, the condition above is equivalent to saying that either the mean-payoff of $s$ (when ($\sigma, \tau^*$) is played) is non-negative or otherwise it is at most $-D$. 

We define the penalty of $s$, denoted by $P_{G, w}(s)$, as the supremum\footnote{We need to take the supremum here to include the case that $ s $ has penalty of at least $ D $ for every real $ D $. In this case, $P_{G, w}(s) =\infty $ ($P_{G, w}(s)$ will not be well-defined if we use maximum instead of supremum).} of all $D$ such that $s$ has a penalty of at least $D$.
We say that the graph $(G, w)$ has a penalty of at least $D$ if every node $s$ has a penalty of at least $D$, and define $P(G, w)=\min_{s\in G} P_{G, w}(s)$. Note that for any graph $(G, w)$, $P(G, w)\geq 1/n$ since for any cycle $C$, $\sum_{(u, v) \in C} w (u, v) / |C|$ is either non-negative or at most $-1/n$.

\section{Value Iteration Algorithm with Admissible List}\label{sec:value_iteration_algorithm}

In this section we present a variant of the {\em Value Iteration Algorithm} for computing the minimal energies of Brim et al.~\cite{Brim11}. In addition to the graph $(G, w)$, our algorithm uses one more parameter $ A $ which is a sorted list containing all possible minimal energy values. That is, the algorithm is promised that $ \en^*_{G, w} (v) \in A $ for every node $ v $.
We call any sorted list $A$ such that $ e^*_{G, w} (v) \in A $ for every node $ v $ an {\em admissible list}.
We show the following proposition.

\setcounter{ValueIterationCounter}{\value{theorem}}
\begin{proposition}\label{pro:value_iteration}
There is an algorithm that, given a (sorted) admissible list $A$, computes the minimal energies of all nodes in $ (G, w) $ in $ O (m|A|) $ time.
\end{proposition}

In general, the simplest choice of an admissible list is $ A = \{0, 1, \ldots, nW, \infty \} $. In this case the algorithm works like the current fastest pseudopolynomial algorithm by Brim et al.~\cite{Brim11} and has a running time of $O(mnW)$. However, for some natural cases, we can give smaller admissible lists. Our first example are graphs where every edge weight is a multiple of an integer $ B > 0 $, as shown in the following corollary.
This corollary will be used later in this paper.

\begin{corollary}\label{cor:algorithm_rounded_game}
Let $ (G, w) $ be a graph for which there is an integer $ B > 0 $ such that the weight of every edge $ (u, v) \in G $ is of the form $ w(u,v) = i B $ for some integer $ i $, and $M$ is an upper bound on the finite minimal energy (i.e., for any node $v$, if $\en^*_{G, w}(v)<\infty$, then $\en^*_{G, w}(v)\leq M$).
There is an admissible list of size $O(M/B)$ which can be computed in $O(M/B)$ time. Thus there is an algorithm that computes the minimal energies of $ (G, w) $ in $ O (m M/B ) $ time.
\end{corollary}

Our second example are graphs in which we have a (small) set of values $ \{ w_1, \ldots, w_d \} $ of size~$ d $ and a window size~$ \delta $ such that \emph{every} weight lies in $ \{ w_i - \delta, \ldots, w_i+\delta \} $ for one of the values $ w_i $.
This is exactly the situation described in Theorem~\ref{thm:algorithm_windowed_game}. Since we prove this theorem in this section, we restate it here. 

\setcounter{theoremtmp}{\value{theorem}}
\setcounter{theorem}{\value{TheoremThree}}
\begin{theorem}[Restated]
If there are $ d $ values $ w_1, \ldots, w_d $ and a window size $ \delta $ such that for every edge $ (u, v) \in G $ we have $ w(u, v) \in \{ w_i - \delta, \ldots, w_i + \delta \} $ for some $ 1 \leq i \leq d $, then the minimal energies can be computed in $ O (m \delta n^{d+1}) $ time.
\end{theorem}
\setcounter{theorem}{\value{theoremtmp}}

As noted in Section~\ref{sec:intro}, in some applications $ d $ is a constant and $ \delta $ is polynomial in $ n $. In this case Theorem~\ref{thm:algorithm_windowed_game} implies a polynomial-time algorithm.

In the rest of this section we first give a proof of Proposition~\ref{pro:value_iteration} (cf. Section~\ref{sec:proofs_value_iteration}).
We subsequently use it to prove Corollary~\ref{cor:algorithm_rounded_game}, and Theorem~\ref{thm:algorithm_windowed_game} (cf. Section~\ref{sec:proofs_admissible_list}).
In both cases, we first highlight the main ideas before giving the full proofs.

\subsection{Proof of Proposition~\ref{pro:value_iteration}}\label{sec:proofs_value_iteration}

In the following we describe the modified value iteration algorithm for computing minimal energies and prove its correctness and running time as stated in Proposition~\ref{pro:value_iteration}.
The value iteration algorithm relies on the following characterization of the minimal energy.

\begin{lemma}[Minimal Energy Characterization \cite{Brim11}]\label{lem:minimal characterization}
An energy function $ \en $ is the minimal energy function of a weighted graph $ (G, w) $ if and only if it fulfills the following three conditions:
\begin{enumerate}
\item \label{item:one} For every node $ u \in V_A $, $ \en (u) + w (u, v) \geq \en (v) $ for {\em some} edge $ (u, v) \in E $.
\item \label{item:two} For every node $ u \in V_B $, $ \en (u) + w (u, v) \geq \en (v) $ for {\em every} edge $ (u, v) \in E $.
\item For every energy function $ \en' $ that fulfills conditions \ref{item:one} and \ref{item:two} we have $ \en(v) \leq \en'(v) $ for every node $ v \in V $.
\end{enumerate}
\end{lemma}

Note that the first two conditions of this lemma are trivially satisfied for a node $ u $ if we set $ \en (u) = \infty $.
An intuitive interpretation of the first two conditions is this: Consider any node $u$ of Alice. 
If we believe that $\en(v)$ is sufficient for all neighbors $v$ of $u$, then $\en(u)$ should be sufficient if, when the car has energy $\en(u)$ at $u$, she can move the car to some neighboring node $v$ to make sure that the energy of the car is still sufficient, i.e., $\en(u)+w(u, v)\geq \en(v)$. Similarly, if $u$ is Bob's node and we believe that $\en(v)$ is sufficient for all neighbors $v$ of $u$, then $\en(u)$ should be sufficient if, when the car has energy $\en(u)$ at $u$, it can be guaranteed that the energy is still sufficient for any neighbor $v$ the car is moved to, i.e., $\en(u)+w(u, v)\geq \en(v)$ for all $v$.

The first two conditions give a sufficient condition for an energy function to be sufficient. It can be shown that these conditions are {\em not} necessary (i.e., some sufficient energy functions do not satisfy these conditions). However, an interesting property of these conditions is that it {\em is} necessary for an energy to be minimal.
Since there could be non-minimal energy functions that satisfy the first two conditions, we have to add the third condition: at {\em all} nodes, the minimal energy function has to be smaller than {\em all} other functions that satisfy the first two conditions.
All three conditions together characterize the (unique) minimal energy function. 

We will first give a general algorithm based on value iteration, called Algorithm~\ref{alg:modified_value_iteration_algorithm}, in which nodes are ``updated'' in an arbitrary order.
We will prove the correctness of this algorithm.
Then we will present a second, faster algorithm, called Algorithm~\ref{alg:modified_value_iteration_algorithm_speed_up}, that processes the nodes in a specific order and that uses a simple data structure.
We will argue that this algorithm gives the desired running time.

The basic idea of Algorithm~\ref{alg:modified_value_iteration_algorithm} is as follows. The algorithm starts with an energy function $ \en(v) = \min A$ for every node $ v $ and keeps increasing $ \en $ slightly in an attempt to satisfy the first two conditions in Lemma~\ref{lem:minimal characterization}.
That is, as long as these conditions are not fulfilled for some node $ u $, it increases $ \en(u) $ to the next value in $A$, which could also be $\infty$. This updating process is repeated until $ \en$ satisfies the conditions (which will eventually happen at least when all $\en(u)$ become $\infty$).
This updating process of the algorithm is the same as in the algorithm of Brim et al. except that $ e(u) $ always increases to the next value in $ A $ and not only to the value given by Lemma~\ref{lem:minimal characterization}.

%\begin{proof}[Proof of Proposition~\ref{pro:value_iteration}]
\paragraph{Correctness.}
Algorithm~\ref{alg:modified_value_iteration_algorithm} shows a simplified version of the algorithm.
We adapt the correctness proof of Brim et al.~\cite{Brim11} to our notation.
It turns out that our modification of the algorithm using a list of admissible values does not disturb the overall correctness argument.

\begin{algorithm}
\caption{Modified value iteration algorithm}
\label{alg:modified_value_iteration_algorithm}

\DontPrintSemicolon
\LinesNotNumbered

\KwIn{A weighted graph $ (G, w) $, a sorted list $ A $ of admissible values for the minimal energies}
\KwOut{The minimal energy of $ (G,w) $}

\BlankLine
\nl $ \en (u) \gets \min A  $ for every $ u \in V $ \tcp*[r]{Initialization}
\tcp{Repeat as long as some node $ u $ violates the first two conditions of Lemma~\ref{lem:minimal characterization}}
\nl \While{there is a node $ u \in V $ such that $ u \in V_A $ and \\
	\mbox{}\phantom{\textbf{while} } $( u \in V_A $ and $ \forall (u, v) \in E: \en(u) + w(u, v) < \en(v) )$ or \\
	\mbox{}\phantom{\textbf{while} } $( u \in V_B $ and $ \exists (u, v) \in E: \en(u) + w(u, v) < \en(v) )$}{
	\BlankLine
	\tcp{Update $ \en(u) $}
\nl 	\uIf{$ u \in V_A $}{
\nl 		$ \en (u) \gets \min_{(u, v) \in E} (\en(v) - w(u, v)) $\;
 	}
\nl 	\ElseIf{$ u \in V_B $}{
\nl 		$ \en (u) \gets \max_{(u, v) \in E} (\en(v) - w(u, v)) $\;
 	}
	\BlankLine
\nl 	\tcp{Increase $ \en (u) $ to next admissible value}
\nl 	$ \en (u) \gets \min \{r \in A \mid r \geq \en(u) \} $\label{lin:update_to_admissible_value}
}
\nl \KwRet{$ \en $}
\end{algorithm}

We first prove the following invariant: after every iteration of the algorithm we have $ \en (x) \leq \en^*_{G, w} (x) $ for every node $ x $.
The statement is certainly true before the first iteration: Since $ \en^*_{G, w} (x) \in A $ we have $ \min A \leq \en^*_{G, w} (x) $.

Now assume that $ \en (x) \leq \en^*_{G, w} (x) $ for every node $ x $ at the beginning of the current iteration.
Let $ u $ be the node that is updated in the current iteration.
For every node $ x \neq u $ the value of $ \en (x) $ does not change in the current iteration.
Let $ \en' (u) $ be the value \emph{before} the energy of $ u $ is increased to the next admissible value in Line~\ref{lin:update_to_admissible_value} and let $ \en'' (u) $ be the value after this operation.
Since $ \en'' (u) = \min \{ r \in A \mid r \geq \en' (u) \} $, it is sufficient to show that $ \en' (u) \leq \en^*_{G, w} (u) $.

Consider first the case that $ u \in V_A $.
In this case we have $ \en (u) + w (u, y) < \en (y) $ for every edge $ (u, y) $ because otherwise the algorithm would not update $ u $.
After the update (and before the execution of Line~\ref{lin:update_to_admissible_value}) we still have $ \en' (u) + w (u, y) \leq \en (y) $ for every edge $ (u, y) $.
Since $ \en^*_{G, w} $ is the minimal energy function we have $ \en^*_{G, w} (u) + w(u, v) \geq \en^*_{G, w} (v) $ for some edge $ (u, v) $ by Lemma~\ref{lem:minimal characterization}.
By the induction hypothesis we have $ \en (v) \leq \en^*_{G, w} (v) $.\footnote{Remember that we assume that there are no self-loops and therefore $ v \neq u $.}
Therefore we get
\[
\en' (u) + w (u, v) \leq \en (v) \leq \en^*_{G, w} (v) \leq \en^*_{G, w} (u) + w (u, v)
\]
and it follows that $ \en' (u) \leq \en^*_{G, w} (u) $.

Consider now the case that $ u \in V_B $.
In this case we have $ \en (u) + w (u, v) < \en (v) $ for at least one edge $ (u, v) $.
After the update (and before the execution of Line~\ref{lin:update_to_admissible_value}) we still have $ \en' (u) + w (u, v) = \en (v) $ for at least one edge $ (u, v) $.
Since $ \en^*_{G, w} $ is the minimal energy function we have $ \en^*_{G, w} (u) + w(u, y) \geq \en^*_{G, w} (y) $ for every edge $ (u, y) $ by Lemma~\ref{lem:minimal characterization}.
In particular this holds for the edge $ (u, v) $.
By the induction hypothesis we have $ \en' (v) = \en (v) \leq \en^*_{G, w} (v) $.
In total we get
\[
\en' (u) + w (u, v) = \en (v) \leq \en^*_{G, w} (v) \leq \en^*_{G, w} (u) + w (u, v)
\]
and it follows that $ \en' (u) \leq \en^*_{G, w} (u) $.

This concludes the proof that for the energy function $ \en $ returned by our algorithm we have $ \en (x) \leq \en^*_{G, w} (x) $ for every node $ x $.
Clearly, the energy function returned by our algorithm fulfills the first two conditions of Lemma~\ref{lem:minimal characterization} because otherwise it would not have terminated.
Thus, our algorithm returns the minimal energies, i.e., $ \en (v) = \en_{G, w} (v) $ for every node $ v $.
We remark that the order in which the nodes are processed in the while loop is irrelevant for the correctness proof.
We will use this fact in the following improved algorithm, Algorithm~\ref{alg:modified_value_iteration_algorithm_speed_up}.

\paragraph{Running Time.}
A running time of $ O(m n |A|) $ for Algorithm~\ref{alg:modified_value_iteration_algorithm} is immediate as every node has to be updated at most $ |A| $ times and both updating a node and checking whether it has to be updated takes time proportional to its out-degree.
The speed-up technique of Brim et al.~\cite{Brim11} also works for our modification and gives a running time of $ O(m |A|) $.
The idea is to maintain a counter for Alice's nodes that keeps track of the number of outgoing edges which fulfill the first condition of Lemma~\ref{lem:minimal characterization}.
The energy only has to be updated if the counter reaches $ 0 $.
Algorithm~\ref{alg:modified_value_iteration_algorithm_speed_up} is the full algorithm which we show for the sake of completeness.

\begin{algorithm}
\caption{Modified value iteration algorithm with speed-up technique}
\label{alg:modified_value_iteration_algorithm_speed_up}

\DontPrintSemicolon

\KwIn{A weighted graph $ (G, w) $, a sorted list $ A $ of admissible values for the minimal energies}
\KwOut{The minimal energy of $ (G,w) $}

\BlankLine
\tcp{Initialization}
$ L \gets \{ u \in V_A \mid \forall (u, v) \in E: \en(u) + w(u, v) < \en(v) \} $\label{lin:initialization_start}\;
$ L \gets \{ u \in V_B \mid \exists (u, v) \in E: \en(u) + w(u, v) < \en(v) \} \cup L $\;
$ \en (u) \gets \min A  $ for every $ u \in V $\;
$ \counter (u) \gets 0 $ for every $ u \in V_A \cap L $\;
$ \counter (u) \gets | \{ v \in V \mid (u, v) \in E, \en(u) + w(u, v) \geq \en(v) \} | $ for every $ u \in V_A \setminus L $\label{lin:initialization_end}\;

\BlankLine
\tcp{Repeat as long as some node $ u $ violates the first two conditions of Lemma~\ref{lem:minimal characterization}}
\While{$ L \neq \emptyset $}{
	Pick $ u \in L $\;
	$ L \gets L \setminus \{ u \} $\;
	$ e_\text{old} \gets \en (u) $\;

	\BlankLine
	\tcp{Update node $ u $}
	\uIf{$ u \in V_A $}{\label{lin:start_update_mechanism}
		$ \en (u) \gets \min_{(u, v) \in E} (\en(v) - w(u, v)) $\;
	}
	\ElseIf{$ u \in V_B $}{
		$ \en (u) \gets \max_{(u, v) \in E} (\en(v) - w(u, v)) $\;
	}
	$ \en (u) \gets \min \{r \in A \mid r \geq \en(u) \} $\label{lin:end_update_mechanism}\;
	\If{$ u \in V_A $}{
		$ \counter (u) \gets | \{ v \in V \mid (u, v) \in E, \en(u) + w(u, v) \geq \en(v) \} | $\;
	}
	
	\BlankLine
	\tcp{Check whether neighbors of $ u $ have to be updated}
	\ForEach{$ t \in V $ such that $ (t, u) \in E $ and $ \en(t) + w(t, u) < \en(u) $}{
		\uIf{$ t \in V_A $}{
			\If{$ \en(t) + w(t, u) \geq e_\text{old} $}{
				$ \counter (t) \gets \counter (t) - 1 $\;
			}
			\If{$ \counter (t) \leq 0 $}{
				$ L \gets L \cup \{ t \} $\;
			}
		}
		\ElseIf{$ t \in V_B $}{
			$ L \gets L \cup \{ t \} $\;
		}
	}
}
\KwRet{$ \en $}
\end{algorithm}

To show the correctness of this algorithm the following two invariants are needed:
\begin{enumerate}
\item For every node $ u \in V \setminus L $ the following holds:
\begin{itemize}
\item If $ u \in V_A $, then there is an edge $ (u, v) $ such that $ \en (u) + w (u, v) \geq \en (v) $
\item If $ u \in V_B $, then for every edge $ (u, v) $ we have $ \en (u) + w (u, v) \geq \en (v) $
\end{itemize}
\item If $ u \in V_A \setminus L $, then $ \counter (u) = | \{ v \in V \mid (u, v) \in E, \en(u) + w(u, v) \geq \en(v) \} | $
\end{enumerate}
The proof of these invariants does not differ from the one given by Brim et al.~\cite{Brim11} which is why we omit it here.
The update mechanism in Lines~\ref{lin:start_update_mechanism} to~\ref{lin:end_update_mechanism} is the same as in Algorithm~\ref{alg:modified_value_iteration_algorithm} which we already proved to be correct.

We now obtain the desired running time of Algorithm~\ref{alg:modified_value_iteration_algorithm} as follows.
For every node $ u $, we let $ \degree^+ (u) $ and $ \degree^- (u) $ denote its out-degree and in-degree, respectively.
The initialization steps in Lines~\ref{lin:initialization_start} to \ref{lin:initialization_end} of the algorithm need time $ O(\degree^+ (u)) $ for every node $ u $.
Thus, the total initialization cost is $ O (\sum_{u \in U} \degree^+ (u)) = O(m) $.
Each iteration of the while loop in which we update a node $ u $ needs time $ O (\degree^+ (u) + \degree^- (u)) $.
Since the energy of every node can increase at most $ |A| $ times, the total running time of this Algorithm~\ref{alg:modified_value_iteration_algorithm_speed_up} is
\[
O \left( \sum_{u \in V} (\degree^+ (u) + \degree^- (u)) \cdot |A| \right) = O (m |A|)\,.
\]
This completes the proof of Proposition~\ref{pro:value_iteration}. 
%\end{proof}

\subsection{Proofs of Corollary~\ref{cor:algorithm_rounded_game} and Theorem~\ref{thm:algorithm_windowed_game}}\label{sec:proofs_admissible_list}

We now prove that in the two special cases described in Corollary~\ref{cor:algorithm_rounded_game} and Theorem~\ref{thm:algorithm_windowed_game}, we can give explicit formulations of admissible lists.
For both proofs we first characterize what values the minimal energy can assume, dependent on the set of edge weights and an upper bound $ M $ on the finite minimal energy.
Specifically, we define
\begin{equation}\label{eq:U M}
U_M = \{ 0, \ldots, M, \infty\} \, .
\end{equation}
We denote the set of different weights of a graph $ (G, w) $ by
\begin{equation}\label{eq:R G w}
R_{G, w} = \{ w(u, v) \mid (u, v) \in E \} .
\end{equation}
The set of all (negated) combinations of edge weights is defined as
\begin{equation}\label{eq:C G w}
C_{G, w} = \left\{ - \sum_{i=1}^k x_i \mid x_i \in R_{G, w} \text{ for all $ i $}, 0 \leq k \leq n \right\} \cup \{\infty\}\,.
\end{equation}
Our key observation is the following lemma.
\begin{lemma}\label{lem:universal_admissible_list}
For every graph $ (G, w) $ with an upper bound $ M $ on the finite minimal energy we have $ \en^*_{G, w} (v) \in C_{G, w}\cap U_M $ for every node $ v \in V $.
\end{lemma}

\begin{proof}
If $ \en^*_{G, w} (v) = \infty $ then we clearly have $ \en^*_{G, w} (v) \in C_{G, w}\cap U_M $.
If $ \en^*_{G, w} (v) < \infty $ we have $ \en^*_{G, w} (v) \in U_M $ since $ M $ is an upper bound on the finite minimal energy.
We still have to show that $ \en^*_{G, w} (v) \in C_{G, w} $.

Let $ (\sigma^*, \tau^*) $ be a pair of optimal strategies.
Since $ \sigma^* $ and $ \tau^* $ are optimal we have $ \en^*_{G, w} (v) = \en^*_{G(\sigma^*, \tau^*), w} (v) < \infty $.
By the definition of the minimal energy (see Section~\ref{sec:definition_energy}) we have
\begin{equation*}
\en^*_{G(\sigma^*, \tau^*), w} (v) = \max \{0, - \min_{P} w(P)\}
\end{equation*}
where the minimization is over all simple paths in $ G(\sigma^*, \tau^*) $ originating at $ v $ and $ w(P) $ denotes the sum of the edge weights of the path $ P $.
If $ \en^*_{G, w} (v) = 0 $ we have $ \en^*_{G, w} (v) \in C_{G, w} $ by setting $ k = 0 $ (the empty sum has value $ 0 $).
Otherwise we have
\begin{equation*}
\en^*_{G, w} (v) = - \sum_{(x, y) \in P} w(x, y)
\end{equation*}
for some simple path $ P $ in $ G(\sigma^*, \tau^*) $ originating at $ v $.
Since the length of $ P $ is at most $ n $ we have at most $ n $ edges on $ P $  which makes it clear that $ \en^*_{G, w} (v) \in C_{G, w} $.
\end{proof}

%\begin{proof}[Proof of Corollary~\ref{cor:algorithm_rounded_game}]
\paragraph{Proof of Corollary~\ref{cor:algorithm_rounded_game}.}
We want to use the value iteration algorithm of Proposition~\ref{pro:value_iteration} with the list
\begin{equation*}
A = \left\{ i \cdot B \mid 0 \leq i \leq \left\lceil \frac{M}{B} \right\rceil \right\} \cup \{ \infty \} \, .
\end{equation*}
It is clear that $ A $ has size $ O(M/B) $ and can be generated in $ O(M/B) $ time.
Thus, we only have to show that $ A $ is admissible to apply Proposition~\ref{pro:value_iteration}.

We will now show that  $ C_{G, w} \cap U_M \subseteq A $ where  $ C_{G, w}$ and $U_M$ are as in Lemma~\ref{lem:universal_admissible_list}. Let $ y \in C_{G, w} \cap U_M $.
The set of different edge weights is $ R_{G, w} \subseteq \{ i \cdot B \mid -W/B \leq i \leq  W/B \} $.
Since $ y \in C_{G, w} $ there is some $ k $ ($ 0 \leq k \leq n $) such that
\begin{equation*}
y = - \sum_{j=1}^k x_j
\end{equation*}
where $ x_j \in R_{G, w} $ for every $ 1 \leq j \leq k $.
Therefore there is an integer $ i_j $ for every $ 1 \leq j \leq k $ such that $ x_j = i_j B $ and we get
\begin{equation*}
y = - \sum_{j=1}^k i_j B = - B \sum_{j=1}^k i_j = - i B
\end{equation*}
for some integer $ i $.
Since $ y \in U_M $ we have $ 0 \leq -i B \leq M $ and therefore $ 0 \leq -i \leq M/B \leq \lceil M/B \rceil $.
Thus, $ y = -i B \in A $ which proves $ C_{G, w} \cap U_M \subseteq A $.
Since $ C_{G, w} \cap U_M $ is admissible by Lemma~\ref{lem:universal_admissible_list} also $ A $ is admissible, i.e., $ e^*_{G, w} (v) \in A $ for every node $ v $. This completes the proof of Corollary~\ref{cor:algorithm_rounded_game}.
%\end{proof}

%\begin{proof}[Proof of Theorem~\ref{thm:main 3}]
\paragraph{Proof of Theorem~\ref{thm:main 3}.}
We want to use the value iteration algorithm of Proposition~\ref{pro:value_iteration} with the list
\begin{equation*}
A' = \left\{ x - \sum_{j=1}^k w_{i_j} \mid 1 \leq i_j \leq d, 0 \leq k \leq n, -n \delta \leq x \leq n \delta \right\} \cup \{\infty\} \, .
\end{equation*}
To show Theorem~\ref{thm:main 3} we have to prove three things:
\begin{enumerate}
\item $ A' $ is an admissible list.
\item $ A' $ has size $ O(\delta n^{d+1}) $.
\item A sorted version of $ A' $ can be computed in $ O(\delta n^{d+1}+dn^{d}\log n) $ time.
\end{enumerate}

We will now show that $ C_{G, w} \subseteq A' $ where $ C_{G, w}$ is as in Lemma~\ref{lem:universal_admissible_list}. Let $ y \in C_{G, w} $.
By the definition of $ C_{G, w} $ there is some $ k $ ($ 0 \leq k \leq n $) such that there are $ k $ edge weights $ x_1, \ldots, x_k \in R_{G, w} $ such that
\begin{align*}
y = - \sum_{j=1}^k x_j \, .
\end{align*}
By the structure of $ R_{G, w} $, the set of all edge weights, we have, for every $ 1 \leq j \leq k $, $ x_j = w_{i_j} + \delta_j $ for some $ i_j $ and $ \delta_j $ such that $ 1 \leq i_j \leq d $ and $ -\delta \leq \delta_j \leq \delta $ which gives
\begin{align*}
y = - \sum_{j=1}^k (w_{i_j} + \delta_j)
\end{align*}
Now observe that
\begin{align*}
- \sum_{j=1}^k (w_{i_j} + \delta_j) = - \sum_{j=1}^k w_{i_j} - \sum_{j=1}^k \delta_j = x - \sum_{j=1}^k w_{i_j} \, .
\end{align*}
for some $ x $ such that $ - n \delta \leq - k \delta \leq x \leq k \delta \leq n \delta $.
Therefore $ y \in A' $ which proves that $ C_{G, w} \subseteq A' $.
Since $ C_{G, w} $ is admissible by Lemma~\ref{lem:universal_admissible_list}, also $ A' $ is admissible.

We now consider the size of $ A' $.
We define
\begin{equation*}
S = \left\{ - \sum_{j=1}^k w_{i_j} \mid 1 \leq i_j \leq d \text{ for all $ j $}, 0 \leq k \leq n, \right\}
\end{equation*}
and get that 
$$ A' = \left\{ y + x \mid y \in S, -n \delta \leq x \leq n \delta \right\}\cup\{\infty\} .$$
We now bound the size of $ S $ as follows.
Each element of $ S $ is a sum of at most $ n $ numbers, each chosen from $ \{w_1, \ldots, w_d \} $.
Therefore, such an element is of the form $ \sum_{i=1}^{d} n_i w_i $ where each $ n_i $ is chosen from $ \{ 0, 1, \ldots, n \} $.
Thus, the size of $ S $ is $ O(n^d) $ and the size of $ A' $ is $ O(\delta n^{d+1}) $.

For the computation of $ A' $ we first compute $ S $.
Sorting $ S $ takes time $ O(|S| \cdot \log{|S|}) $ which is $ O (d n^d \log{n}) $.
We iterate over every element $ y \in S $ and generate every integer $ i $ in $ [ y - n \delta, y + n \delta ] $.
We append $ i $ to the list $ A' $ if it is larger than the current last element of the list.
Since for every $ y \in S $ the interval that we consider has the same ``width'' of $ n \delta $, it can never happen that we generate an integer $ i $ that is smaller than the last element and does not yet occur in the list.
Therefore $ A' $ is always sorted.
This process takes time $ O(|A'|) = O(\delta n^{d+1}) $.
In total it takes time $ O(\delta n^{d+1}+dn^{d}\log n) $ to compute $ A' $.

By Proposition~\ref{pro:value_iteration} it takes time $ O(m |A'|) $ to compute the minimal energies.
When we add the construction time of $ A' $ we get a total running time of $ O(\delta m n^{d+1}+dn^{d}\log n) $.
Note that it is always possible to group the edge weights into $ d = m $ groups such that every group contains only one edge weight.
Therefore we may assume that $ d \leq m $.
In that case the first term dominates the second term which gives a total running time of $ O (\delta m n^{d+1}) $.
This completes the proof of Theorem~\ref{thm:main 3}.
%\end{proof}

\section{Approximating Minimal Energies for Large Penalties}\label{sec:rounding_procedure}

This section is devoted to proving Theorem~\ref{thm:rounding_procedure_gives_desired_approximation}. We restate it here for convenience.
\setcounter{theoremtmp}{\value{theorem}}
\setcounter{theorem}{\value{TheoremTwo}}
\begin{theorem}[Restated]
Given a graph $(G, w)$ with $ P(G, w) \geq 1 $, an integer $M$, and an integer $c$ such that $n\leq c \leq n P(G, w)$, we can compute an energy function $ \en $ such that
\begin{align*}
\en (v) \leq \en^*_{G, w} (v) \leq \en (v) + c
\end{align*}
for every node $v$ in $O (mnM/c) $ time, provided that for every node $v$, $e_{G, w}^*(v)<\infty$ implies that $e_{G, w}^*(v)\leq M$.
\end{theorem}
\setcounter{theorem}{\value{theoremtmp}}
We show that we can approximate the minimal energy of nodes in high-penalty graphs (see Section~\ref{sec:prelim} for the definition of penalty).
The key idea is {\em rounding edge weights}, as follows. For an integer $ B > 0 $ we denote the weight function resulting from rounding up every edge weight to the nearest multiple of $ B $ by $ w_B $. Formally, the function $ w_B $ is given by
\begin{equation*}
w_B (u, v) = \left\lceil \frac{w(u, v)}{B} \right\rceil \cdot B
\end{equation*}
for every edge $ (u, v) \in E $. Our algorithm is as follows. We set $B = \lfloor c/n \rfloor \leq P (G, w) $ (where $c$ is as in Theorem~\ref{thm:rounding_procedure_gives_desired_approximation}). Since weights in $(G, w_B)$ are multiples of $B$, $\en^*_{G, w_B}$ can be found faster than $\en^*_{G, w}$ due to Corollary~\ref{cor:algorithm_rounded_game}: we can compute $\en^*_{G, w_B}$ in time $ O (m M/B)=O(mnM/c) $ provided that $ M $ is an upper bound on the finite minimal energy. This is the running time stated in Theorem~\ref{thm:rounding_procedure_gives_desired_approximation}. We complete the proof of Theorem~\ref{thm:rounding_procedure_gives_desired_approximation} by showing that $\en^*_{G, w_B}$ is a good approximation of $\en^*_{G, w}$ (i.e., it is the desired function $\en$).
Recall that by the definition of $ P (G, w) $ every node $ v $ has penalty $ P_{G, w} (v) \geq P (G, w) $.
\begin{proposition}\label{prop:bound error}
For every node $ v $ with penalty $ P_{G, w} (v) \geq B=\lfloor c/n \rfloor $ (where $ c \geq n $) we have
\begin{align*}
\en^*_{G, w_B} (v) \stackrel{\text{(1)}}{\leq} \en^*_{G, w}  (v) \stackrel{\text{(2)}}{\leq} \en^*_{G, w_B}  (v) + nB \leq \en^*_{G, w_B}  (v) + c \, .
\end{align*}
\end{proposition}

The rest of this section is devoted to proving Proposition~\ref{prop:bound error}.
\footnote{
At this point we remark that energy games are not as resistant to perturbations of weights as mean-payoff games.
In particular, if $ w (u, v) \leq w' (u, v) \leq w (u, v) + x $ for every edge $ (u, v) $ and some positive constant $ x $, then also $ \val (v) \leq \val' (v)  \leq \val (v) + x $, where $ \val (v) $ and $ \val' (v) $ are the values of the mean-payoff games for $ v $ in $ (G, w) $ and $ (G, w') $, respectively.
A similar inequality is not true for the minimal energies.
Consider a cycle of total weight $ 0 $.
By adding $ -1 $ to each edge weight, the weight of this cycle changes from non-negative to negative.
Thus, the minimal energy might change from $ 0 $ to $ \infty $.
}
Let us first give the proof ideas.
The last inequality in the proposition follows immediately from the definition of $B$. The first two inequalities will be proved in Section~\ref{sec:details_rounding_procedure 1} and \ref{sec:details_rounding_procedure 2}. Let us first outline the proofs of these inequalities here. 
Inequality~(1) is quite intuitive: We are doing Alice a favor by {\em increasing} edge weights from $w$ to $w_B$. Thus, Alice should not require more energy in $(G, w_B)$ than she needs in $(G, w)$. As we show in Lemma~\ref{lem:energy_increased_weights} in Section~\ref{sec:details_rounding_procedure 1}, this actually holds for {\em any} increase in edge weights: For any $w'$ such that $w'(u,v)\geq w(u, v)$ for all $(u, v)\in G$, we have $ \en^*_{G, w'} (v) \leq \en^*_{G, w} (v) $. Thus we get the first inequality by setting $w'=w_B$.

For inequality (2) in Proposition~\ref{prop:bound error}, unlike the first inequality, we do not state this result for general increases of the edge weights as the bound depends on
our rounding procedure. At this point we also need the precondition that the graph we consider has penalty at least $B$.
We first show that the inequality holds when the strategies played by both players fulfill a certain condition, formally stated as follows (we prove this lemma in Section~\ref{sec:details_rounding_procedure 2}).
\newcounter{LemmaOne}
\setcounter{LemmaOne}{\value{theorem}}
\begin{lemma}\label{lem:energy_rounding_strategies}
Let $ (\sigma, \tau) $ be a pair of strategies. For any node $v$,
if $ \en^*_{G (\sigma, \tau), w} (v) = \infty $ implies $ \en^*_{G(\sigma, \tau), w_B} (v) = \infty $, then $ \en^*_{G(\sigma, \tau), w} (v) \leq \en^*_{G(\sigma, \tau), w_B} (v) + nB $.
\end{lemma}
The above lemma needs a pair of strategies $(\sigma, \tau)$ such that $ \en^*_{G (\sigma, \tau), w} (v) = \infty $ implies $ \en^*_{G(\sigma, \tau), w_B} (v) = \infty $.
This property can be explained as follows: If Alice needs infinite energy at node $v$ in the graph $(G(\sigma, \tau), w)$ then she also needs infinite energy in the rounded-weight graph $(G(\sigma, \tau), w_B)$.
Our second crucial fact shows that if $v$ has penalty at least $B$ then there exists a pair of strategies that has this property.
{\em This is where we exploit the fact that the penalty is large.}
\newcounter{LemmaTwo}
\setcounter{LemmaTwo}{\value{theorem}}
\begin{lemma}\label{lem:penalty_implies_condition_for_rounding_lemma}
Let $ v $ be a node with penalty $ P_{G, w} (v) \geq B $.
Then there is an optimal strategy $ \tau^* $ of Bob such that for every strategy $ \sigma $ of Alice we have that $ \en^*_{G(\sigma, \tau^*), w} (v) = \infty $ implies $ \en^*_{G(\sigma, \tau^*), w_B} (v) = \infty $.
\end{lemma}

%%%%%%%%%%%%%%%%%%%%%%%%%%%%%%%%%%%%%%%%%%%%%%%%%%%%%%%%%

To prove Lemma~\ref{lem:energy_rounding_strategies} we only have to consider a special graph where the strategies of both players are fixed and thus all nodes have out-degree one. The challenge in proving Lemma~\ref{lem:penalty_implies_condition_for_rounding_lemma} is to use the ``right'' strategy $\tau^*$.
We use the strategy $\tau^*$ that comes from the definition of the penalty (cf. Section~\ref{sec:prelim}). The full proofs of Lemmas~\ref{lem:energy_rounding_strategies} and \ref{lem:penalty_implies_condition_for_rounding_lemma} are given in Section~\ref{sec:details_rounding_procedure 2}.

The other challenge of the proof of Proposition~\ref{prop:bound error} is translating our result from graphs with fixed strategies to general graphs in order to prove the second inequality in Proposition~\ref{prop:bound error}. We do this as follows. Let $ \sigma^* $ be an optimal strategy of Alice for $ (G, w) $ and let $ (\sigma_B^*, \tau_B^*) $ be a pair of optimal strategies for $ (G, w_B) $.
Since $ v $ has penalty $ P_{G, w} (v) \geq B $, Lemma~\ref{lem:penalty_implies_condition_for_rounding_lemma} tells us that the preconditions of Lemma~\ref{lem:energy_rounding_strategies} are fulfilled.
We use Lemma~\ref{lem:energy_rounding_strategies} and get that there is an optimal strategy $ \tau^* $ of Bob such that $ e^*_{G(\sigma_B^*, \tau^*), w} (v) \leq e^*_{G(\sigma_B^*, \tau^*), w_B} (v) + nB $.
We now arrive at the chain of inequalities
\begin{align*}
e^*_{G, w} (v) &\stackrel{\text{(a)}}{=} e^*_{G(\sigma^*, \tau^*), w} (v) \stackrel{\text{(b)}}{\leq}  e^*_{G(\sigma_B^*, \tau^*), w} (v) \stackrel{\text{(Lem.~\ref{lem:energy_rounding_strategies})}}{\leq} e^*_{G(\sigma_B^*, \tau^*), w_B} (v) + nB \\
&\stackrel{\text{(c)}}{\leq}e^*_{G(\sigma_B^*, \tau_B^*), w_B} (v) + nB \stackrel{\text{(d)}}{=} e^*_{G, w_B} (v) + nB
\end{align*}
that can be explained as follows.
Since $ (\sigma^*, \tau^*) $ and $ (\sigma_B^*, \tau_B^*) $ are pairs of \emph{optimal} strategies, we have (a) and (d).
Due to the optimality we also have $ e^*_{G(\sigma^*, \tau^*), w} (v) \leq e^*_{G(\sigma, \tau^*), w} (v) $ for \emph{any} strategy $ \sigma $ of Alice, and in particular $ \sigma_B^* $, which implies (b).
A symmetric argument gives (c).

\subsection{Proof of the First Inequality of Proposition~\ref{prop:bound error}}\label{sec:details_rounding_procedure 1}

In the following we prove that an increase in edge weights does not increase the minimal energy for any node.
We first prove the claim for the case where we fix the strategies of both players, i.e., on graphs where we have deleted all edges except those corresponding to the strategies of Alice and Bob.
Afterwards we generalize the claim to arbitrary graphs.

\begin{lemma}\label{lem:energy_increased_weights_strategies}
Let $ G $ be a graph and $ w_1 $ and $ w_2 $ be edge weights such that $ w_1 (u, v) \leq w_2 (u, v) $ for every edge $ (u, v) \in G $.
Then, for every pair of strategies $ (\sigma, \tau) $ and every node $ v \in G $, we have $ \en^*_{G(\sigma, \tau), w_1} (v) \geq \en^*_{G(\sigma, \tau), w_2} (v) $.
\end{lemma}

\begin{proof}
Let $v$ be any node and $(\sigma, \tau)$ be any pair of strategies. First, consider the case where $ \en^*_{G(\sigma, \tau), w_2} (v) = \infty $.
Let $ C $ denote the unique cycle reachable from $ v $ in $ G(\sigma, \tau) $.
Since $ \en^*_{G(\sigma, \tau), w_2} (v) = \infty $ we know by the definition of the minimal energy that $ w_2(C) < 0 $ where $ w_2(C) $ denotes the sum of the edge weights of the cycle~$ C $.
By our assumption we have $ w_1(C) \leq w_2(C) < 0 $, meaning that $ \en^*_{G(\sigma, \tau), w_1} (v) = \infty $ which is exactly what our inequality claims.

Next, consider the case where $ \en^*_{G(\sigma, \tau), w_2} (v) < \infty $.
By the definition of the minimal energy (see Section~\ref{sec:definition}) we have
\begin{equation*}
\en^*_{G(\sigma, \tau), w_2} (v) = \max \left\{ 0, - \min_{P} w_2(P) \right\}
\end{equation*}
where the minimization is over all simple paths in $ (G(\sigma, \tau), w_2) $ originating at $ v $ and $ w_2(P) $ denotes the sum of the edge weights of the path $ P $.

In the case where $ \en^*_{G(\sigma, \tau), w_2} (v) = 0 $, we have $ \en^*_{G(\sigma, \tau), w_2} (v) = 0 \leq \en^*_{G(\sigma, \tau), w_1} (v) $.
If $ \en^*_{G(\sigma, \tau), w_2} (v) > 0 $, we have
\begin{equation*}
\en^*_{G(\sigma, \tau), w_2} (v) = - \min_{P} w_2(P) \, .
\end{equation*}
Since $ w_2(P) \geq w_1(P) $ for every path $ P $ we have
\begin{align*}
\en^*_{G(\sigma, \tau), w_2} (v) &= - \min_{P} w_2(P)\\
&\leq - \min_{P} w_1(P) \\
&\leq \max \{0, - \min_{P} w_1(P)\}\\
&= \en^*_{G(\sigma, \tau), w_1} (v) \, . \qedhere
\end{align*}
\end{proof}

It is now straightforward to generalize the previous lemma by applying it to an \emph{optimal} pair of strategies.

\begin{lemma}\label{lem:energy_increased_weights}
Let $ G $ be a graph and $ w_1 $ and $ w_2 $ be edge weights such that $ w_1 (u, v) \leq w_2 (u, v) $ for every edge $ (u, v) \in G $.
Then $ \en^*_{G, w_1} (v) \geq \en^*_{G, w_2} (v) $ for every node $ v $.
\end{lemma}

\begin{proof}
Let $ (\sigma_1^*, \tau_1^*) $ be an optimal pair of strategies for $ (G, w_1) $ and let $ (\sigma_2^*, \tau_2^*) $ be an optimal pair of strategies for $ (G, w_2) $.
Note that $ \en^*_{G (\sigma_1^*, \tau_1^*), w_1} (v) \geq \en^*_{G (\sigma_1^*, \tau), w_1} (v) $ for every strategy $ \tau $ of Bob (since $\tau_1^*$ is Bob's optimal strategy).
We also have $ \en^*_{G (\sigma, \tau_2^*), w_2} (v) \geq \en^*_{G (\sigma_2^*, \tau_2^*), w_2} (v) $ for every strategy $ \sigma $ of Alice.
Together with Lemma~\ref{lem:energy_increased_weights_strategies} we get
\begin{align*}
\en^*_{G, w_1} (v) &= \en^*_{G (\sigma_1^*, \tau_1^*), w_1} (v) \geq \en^*_{G (\sigma_1^*, \tau_2^*), w_1} (v)
\geq \en^*_{G (\sigma_1^*, \tau_2^*), w_2} (v) \\
&\geq \en^*_{G (\sigma_2^*, \tau_2^*), w_2} (v) = \en^*_{G, w_2} (v) \, . \qedhere
\end{align*}
\end{proof}

\subsection{Proof of the Second Inequality of Proposition~\ref{prop:bound error}}\label{sec:details_rounding_procedure 2}

We now complete the proof of the second inequality of Proposition~\ref{prop:bound error}.
We have already proved this inequality right after the statement of Proposition~\ref{prop:bound error}, but our proof assumes Lemmas~\ref{lem:energy_rounding_strategies} and~\ref{lem:penalty_implies_condition_for_rounding_lemma}. 
In this section, we provide the proofs of these two lemmas. 

\setcounter{theoremtmp}{\value{theorem}}
\setcounter{theorem}{\value{LemmaOne}}
\begin{lemma}[Restated]
Let $ (\sigma, \tau) $ be a pair of strategies. For any node $v$,
if $ \en^*_{G (\sigma, \tau), w} (v) = \infty $ implies that $ \en^*_{G(\sigma, \tau), w_B} (v) = \infty $, then $ \en^*_{G(\sigma, \tau), w} (v) \leq \en^*_{G(\sigma, \tau), w_B} (v) + nB $.
\end{lemma}
\setcounter{theorem}{\value{theoremtmp}}

\begin{proof}
Recall that $ w_B $ is defined as the weight function resulting from rounding up every edge weight of $ w $ to the nearest multiple of $ B $, i.e.,
\begin{equation*}
w_B (u, v) = \left\lceil \frac{w(u, v)}{B} \right\rceil \cdot B \, .
\end{equation*}
By this definition we have $ w_B (u, v) \leq w (u, v) + B $ for every edge $ (u, v) \in E $.

If $ \en^*_{G(\sigma, \tau), w} (v) = \infty $, then also $ \en^*_{G(\sigma, \tau), w_B} (v) = \infty $ which trivially makes the inequality $ \en^*_{G(\sigma, \tau), w} (v) \leq \en^*_{G(\sigma, \tau), w_B} (v) + nB $ hold.
We now consider the case where $ \en^*_{G(\sigma, \tau), w} (v) < \infty $.
By the definition of the minimal energy we have
\begin{equation*}
\en^*_{G(\sigma, \tau), w} (v) = \max \left\{0, - \min_{P} w(P) \right\}
\end{equation*}
where the minimization is over all simple paths in $ (G(\sigma, \tau), w) $ originating at~$ v $ and $ w(P) $ denotes the sum of the edge weights of the path $ P $.
In the case where $ \en^*_{G(\sigma, \tau), w} (v) = 0 $, our claimed inequality trivially holds because $ \en^*_{G(\sigma, \tau), w_B} (v) \geq 0 $.
Consider now the second case where $ \en^*_{G(\sigma, \tau), w} (v) > 0 $. In this case, we have
\begin{equation*}
\en^*_{G(\sigma, \tau), w} (v) = - \min_{P} w(P) \, .
\end{equation*}
Every simple path $ P $ has length at most $ n $ and therefore
\begin{equation*}
w_B(P) = \sum_{(u, v) \in P} w_B(u, v) \leq \sum_{(u, v) \in P} (w (u, v) + B) \leq w (P) + nB \, .
\end{equation*}
Thus, we get $ w(P) \geq w_B(P) - nB $ for every simple path $ P $.
We now get
\begin{align*}
\en^*_{G(\sigma, \tau), w} (v) = - \min_{P} w(P) &\leq - \min_{P} (w_B(P) - nB) \\
&= - \min_{P} (w_B(P)) + n B = \en^*_{G(\sigma, \tau), w_B} (v) + nB \, . \qedhere
\end{align*}
\end{proof}

We now show that the precondition of the previous lemma is already implied by our choice of $ B $.

\setcounter{theoremtmp}{\value{theorem}}
\setcounter{theorem}{\value{LemmaTwo}}
\begin{lemma}[Restated]
Let $ v $ be a node with penalty $ P_{G, w} (v) \geq B $.
Then there is an optimal strategy $ \tau^* $ of Bob such that for every strategy $ \sigma $ of Alice we have that $ \en^*_{G(\sigma, \tau^*), w} (v) = \infty $ implies $ \en^*_{G(\sigma, \tau^*), w_B} (v) = \infty $.
\end{lemma}
\setcounter{theorem}{\value{theoremtmp}}

To prove the above lemma, we first prove the following claim.

\begin{claim}
If the average weight of a cycle $ C $ in $ (G, w) $ is at most $ -B $, then $ C $ is a negative cycle in $ (G, w_B) $ with total weight $ w_B(C) < 0 $.
\end{claim}
\begin{proof}
We assume that the average weight of $ C $ in $ (G, w) $ is at most $ -B $, i.e.,
\begin{equation*}
\frac{\sum_{(u, v) \in C} w (u, v)}{|C|} \leq -B \, .
\end{equation*}
Since $ w_B (u, v) < w (u, v) + B $ for every edge $ (u, v) \in E $, we get the following bound for the average weight of $ C $ in $ (G, w_B) $:
\begin{align*}
\frac{\sum_{(u, v) \in C} w_B (u, v)}{|C|} &< \frac{\sum_{(u, v) \in C} (w (u, v) + B)}{|C|} \\
&= \frac{\sum_{(u, v) \in C} w (u, v)}{|C|} + \frac{\sum_{(u, v) \in C} B}{|C|} \\
&\leq - B + \frac{|C| \cdot B}{|C|} = 0 \, .
\end{align*}
Therefore, $ w_B (C) = \sum_{(u, v) \in C} w_B (u, v) < 0 $ which means that $ C $ is a negative cycle in $ (G, w_B) $.
This finishes the proof of the claim.
\end{proof}

We now give the proof of Lemma~\ref{lem:penalty_implies_condition_for_rounding_lemma}.
\begin{proof}[Proof of Lemma~\ref{lem:penalty_implies_condition_for_rounding_lemma}]
By the definition of the penalty we know that there is an optimal strategy~$ \tau^* $ of Bob such that, for every strategy~$ \sigma $ of Alice, if the unique cycle $ C $ reachable from $ v $ in $ G(\sigma, \tau^*) $ has negative total weight $ w(C) < 0 $, then its average weight is at most $ - P_{G, w} (v) \leq -B $ by the definition of $P(G, w)$.
Now let $ \sigma $ be any strategy of Alice and let $ C $ denote the unique cycle $ C $ reachable from $ v $ in $ G(\sigma, \tau^*) $.
Assume that $ \en^*_{G(\sigma, \tau^*), w} (v) = \infty $.
Then we have $ w(C) < 0 $ and thus, by the definition of the penalty, $ C $ has an average weight of at most $ - B $.
By our claim we get that $ C $ is a negative cycle in $ (G, w_B) $ (i.e. $ w_B (C) < 0$) and therefore $ \en^*_{G (\sigma, \tau^*), w_B} (v) = \infty $.
\end{proof}

\section{Exact Solution by Approximation}\label{sec:exact_solution_by_approximation}

We now use our results from the previous sections to prove Theorem~\ref{thm:main 1}.
\setcounter{theoremtmp}{\value{theorem}}
\setcounter{theorem}{\value{TheoremOne}}
\begin{theorem}[Restated]\label{thm:main 1}
Given a graph $(G, w)$ and an integer $M$ we can compute the minimal initial energies of all nodes in
\begin{equation*}
O \left( mn \left( \log\frac{M}{n} \right) \left( \log\frac{M}{n\lceil P(G, w)\rceil} \right) + m \frac{M}{\lceil P(G, w)\rceil} \right)
\end{equation*}
time, provided that for all $v$, $e_{G, w}^*(v)<\infty$ implies that $e_{G, w}^*(v)\leq M$.
\end{theorem}
\setcounter{theorem}{\value{theoremtmp}}

As the first step, we provide an algorithm that computes the minimal energy given a lower bound on the penalty of the graph.
For this algorithm, we show how we can use the approximation algorithm in Section~\ref{sec:rounding_procedure} to find an {\em exact} solution.
\newcounter{LemmaThree}
\setcounter{LemmaThree}{\value{theorem}}
\begin{lemma}\label{lem:approx to exact}
There is an algorithm that takes a graph $(G, w)$, a lower bound~$ D $ on the penalty $ P (G, w) $, and an upper bound $ M $ on the finite minimal energy of $ (G, w) $ as its input and computes the minimal energies of $ (G, w) $ in $ O(m n \log{D} + m \cdot \frac{M}{\lceil D \rceil})$ time.
Specifically, if $ P(G, w) \geq M/(2n) $, we can set $ D = M/(2n) $ and the algorithm runs in time $ O(m n \log{(M/n)}) $.
\end{lemma}
\begin{proof} [Main Idea]
We provide the main idea of the proof of Lemma~\ref{lem:approx to exact}. Details are in Section \ref{apx:details_exact_solution_by_approximation} and \ref{sec:full proof of lem:approx to exact}. 

To illustrate the main idea, we focus on the case $ D = M / (2n) $ where we want to show an $O(mn\log (M/n))$ running time.
If that condition does not hold, we can transform the problem into a problem where it holds in time $ O (m M/D) $.
Let $\cA$ be the approximation algorithm given in Theorem~\ref{thm:rounding_procedure_gives_desired_approximation}. Recall that $\cA$ takes $c$ as its input and returns $\en(v)$ such that
\begin{align}
\en (v) \leq \en^*_{G, w} (v) \leq \en (v) + c
\end{align}
provided that $n\leq c \leq n P(G, w) $. Our exact algorithm will run $\cA$ with parameter $ c = \lfloor M/2 \rfloor $ which satisfies $ c \leq M/2 \leq n D \leq n P(G, w) $. By Theorem~\ref{thm:rounding_procedure_gives_desired_approximation}, this takes $O(mnM/c)=O(mn)$ time.
Using the energy function $ e $ returned by $\cA$, our algorithm produces a new graph $(G, w')$ defined by $w'(u, v)=w(u, v)+\en(u)-\en(v)$ for every edge $ (u, v)\ $. It can be proved that this graph has the following crucial properties (see details in Lemma~\ref{lem:main_lemma_for_correctness} in Section~\ref{apx:details_exact_solution_by_approximation}):
\begin{enumerate}
\item The penalty does not change, i.e., $ P_{G, w} (v) = P_{G, w'} (v) $ for every node $ v $.
\item We have $\en^*_{G, w}(v)=\en^*_{G, w'}(v)+\en(v)$ for every node $ v $.
\item The largest finite minimal energy of nodes in $ (G, w') $ is at most $c$; i.e., if $ \en^*_{G, w'}(v) < \infty $ then $ \en^*_{G, w'}(v) \leq c $. (This follows from property~2 and the inequality $ \en^*_{G, w} (v) \leq \en (v) + c $ of Theorem~\ref{thm:rounding_procedure_gives_desired_approximation}.)
\end{enumerate}
The algorithm then recurses on input $(G, w')$, $D$ and $ M' = c = \lfloor M/2 \rfloor $. Properties~1 and~3 guarantee that the preconditions of our algorithm for the recursive call are fulfilled:
By our choice of $ M' $ we know that if $ \en^*_{G, w'}(v) < \infty $ then $ \en^*_{G, w'}(v) \leq M' $ and since $ D \leq P_{G, w} (v) = P_{G, w'} (v) $, $ D $ is a lower bound on the penalty of $ (G, w') $.
Therefore we may recurse and the algorithm will return $\en^*_{G, w'}(v)$ for every node $ v $. It then outputs $\en^*_{G, w'}(v)+\en(v)$ which is guaranteed to be a correct solution (i.e., $\en^*_{G, w}(v)=\en^*_{G, w'}(v)+\en(v)$) by the second property.
The running time of this algorithm is $T(n, m, M) \leq T(n, m, M/2)+O(mn)$.
We stop the recursion when $M$ becomes small enough, i.e. when $ M \leq n $. In this case the value iteration algorithm $\cA$ runs in $ O(mn) $ time.
Thus we get $T(n, m, M)=O(mn\log (M/n))$ as desired.
\end{proof}

We now prove Theorem~\ref{thm:main 1} by extending the algorithm of Lemma~\ref{lem:approx to exact} to an algorithm that does not require the knowledge of a lower bound of the penalty.
\begin{proof}[Proof of Theorem~\ref{thm:main 1}]
We repeatedly guess a lower bound for the penalty $ P_{G, w} $ and run the algorithm of Lemma~\ref{lem:approx to exact} until our guess eventually turns out to be correct.
We start with the guess $ D = M / (2n) $ for which the algorithm of Lemma~\ref{lem:approx to exact} runs in time $ O (m n \log(M/n) ) $.
We then perform binary search for the next values of $ D $ by trying the values $ M/(2n) $, $ M/(4n) $, $ M/(8n) $, and so on.

If our guess was correct, the algorithm returns the minimal energy function.
If our guess was not correct, the energy function returned by our algorithm might not necessarily be the minimal energy function.
Using the following characterization of the minimal energy we can check in linear time whether we have already found the minimal energy function.
\begin{lemma}[Minimal Energy Characterization \cite{Lifshits07}]\label{lem:characterization_of_minimal_energy}
The minimal energy of the graph $ (G, w) $ is the unique energy function $ \en $ satisfying
\begin{align*}
\en (u) &= \begin{cases}
\min_{(u, v) \in E} \max (\en (v) - w(u, v), 0) & \text{if $ u \in V_A $} \\
\max_{(u, v) \in E} \max (\en (v) - w(u, v), 0) & \text{if $ u \in V_B $}\, .
\end{cases}
\end{align*}
for every node $ u \in G $.
\end{lemma}
By checking the equation for every node $ u $ we can determine in time $ O (m) $ whether an energy function $ e $ is indeed the minimal energy function.

We stop if we have already found the minimal energy function by running the algorithm of Lemma~\ref{lem:approx to exact} with our guessed lower bound $ D $ of the penalty. Otherwise we guess a new lower bound $ D $ of the penalty which is half of the previous one and run the algorithm of Lemma~\ref{lem:approx to exact} again. Eventually, our guess will be correct and we will stop before the guessed value is smaller than $P(G,w)/2$ or $ 1 $ (in the latter case we simply run the value iteration algorithm). Therefore we get a running time of
\begin{equation*}
\textstyle
O\left( mn \left(\log\frac{M}{2n}+\log\frac{M}{4n}+\ldots+\log(\lceil P(G, w)\rceil) \right)+ m \left( 2n+4n+\ldots+\frac{M}{\lceil P(G, w)\rceil} \right) \right)
\end{equation*}
which solves to $O(mn(\log\frac{M}{n})(\log{\frac{M}{n \lceil P(G, w)\rceil}})+\frac{mM}{\lceil P(G, w)\rceil})$.
\end{proof}
In the worst case, i.e., when $P(G, w) = 1/n$ and $ M = nW $, our algorithm runs in time $ O(mnW) $ which matches the current fastest pseudopolynomial algorithm~\cite{Brim11}.
The result also implies that graphs with a penalty of at least $W/\poly(n)$ form an interesting class of polynomial-time solvable energy games.

\subsection{Auxiliary Lemma Needed for Proving Lemma~\ref{lem:approx to exact}}\label{apx:details_exact_solution_by_approximation}

In the following we prove an auxiliary lemma that we need for arguing about the correctness of the algorithm of Lemma~\ref{lem:approx to exact}.
In that algorithm we first compute an energy function $ \en $ that approximates the minimal energy function of a weighted graph $ (G, w) $ and then define a new weight function $ w' $ by $ w' (u, v) = w (u, v) + \en (u) - \en (v) $ for every edge $ (u, v) $.
For our algorithm to be correct we need two properties to hold.\footnote{The third property we mentioned above follows from property 2 and the approximation guarantee of the energy function $ \en $.}
\begin{enumerate}
\item The penalty does not change, i.e., $ P_{G, w} (v) = P_{G, w'} (v) $ for every node $ v $.
\item We have $\en^*_{G, w}(v)=\en^*_{G, w'}(v)+\en(v)$ for every node $ v $.
\end{enumerate}
We will show that these two properties actually hold for \emph{any} energy function $ \en $.

Note that this kind of modification of the weights is often called a potential transformation~\cite{GurvichKK88} by the potential function $ e $.
It is well-known that a potential transformation does not change the average weight of any cycle and the total weight of a path from $ u $ to $ v $ changes by $ e(u) - e(v) $.
The first property above in fact follows from this observation and we provide its proof only for completeness.
The second property above additionally needs the precondition that $ \en(v) $ does not exceed the minimal energy at $ v $ and is not true for an arbitrary potential transformation.

\begin{lemma}\label{lem:main_lemma_for_correctness}
Let $ (G, w) $ be a weighted graph and let $ \en $ be an energy function such that $ \en (v) \leq \en^*_{G, w} (v) $ for all $ v \in G $.
Define the modified game $ (G, w') $ with the weight function $ w' $ by $ w'(u, v) =  w(u, v) + \en (u) - \en (v) $ for every edge $ (u, v) \in G $.
Then the penalty does not change, i.e., $ P_{G, w} (v) = P_{G, w'} (v) $ for every node $ v \in G $, and 
$ \en^*_{G, w} (v) = \en (v) + \en^*_{G, w'} (v) $ for every node $ v \in G $.
\end{lemma}

\begin{proof}
We first show that the penalty does not change from $ w $ to $ w' $, i.e., $ P_{G, w} = P_{G, w'} $.
For this purpose we will show that every cycle in $ G $ has the same sum of edge weights in $ (G, w) $ and in $ (G, w') $ which means that the average weights are the same.
By the definition of the penalty this implies that $ P_{G, w} (v) = P_{G, w'} (v) $ for every node $ v \in G $ as desired.
Let $ C $ be a cycle of $ G $ consisting of the nodes $ v_1, \ldots, v_k $.
We simply plug in the definition of $ w' $ to check that our claim is true:
\begin{align*}
\sum_{(u, v) \in C} w'(u, v) &= w'(v_k, v_1) + \sum_{i=1}^{k-1} w'(v_i, v_{i+1}) \\
&= w(v_k, v_1) + \en (v_k) - \en (v_1) + \sum_{i=1}^{k-1} \left( w(v_i, v_{i+1}) + \en (v_i) - \en (v_{i+1}) \right) \\
&= w(v_k, v_1) + \en (v_k) - \en (v_1) + \sum_{i=1}^{k-1} w(v_i, v_{i+1}) + \sum_{i=1}^{k-1} \en (v_i) - \sum_{i=2}^{k} \en (v_{i}) \\
&= w(v_k, v_1) + \sum_{i=1}^{k-1} w(v_i, v_{i+1}) + \sum_{i=1}^{k} \en (v_i) - \sum_{i=1}^{k} \en (v_{i}) \\
&= w(v_k, v_1) + \sum_{i=1}^{k-1} w(v_i, v_{i+1}) \\
&= \sum_{(u, v) \in C} w(u, v) \, .
\end{align*}

We now prove the second property.
We define the energy function $ f $ by $ f (v) = \en (v) + \en^*_{G, w'} (v) $ for every node $ u \in G $.
We use Lemma~\ref{lem:characterization_of_minimal_energy} to show that $ f $ is the minimal energy $ \en^*_{G, w} $.
We have to show that, for every node $ u \in G $, we have
\begin{align*}
f (u) &= \begin{cases}
\min_{(u, v) \in E} \max (f (v) - w(u, v), 0) & \text{if $ u \in V_A $} \\
\max_{(u, v) \in E} \max (f (v) - w(u, v), 0) & \text{if $ u \in V_B $}
\end{cases} \, .
\end{align*}
By the definition of $ f $ this is equivalent to
\begin{align*}
\en (u) + \en^*_{G, w'} (u) &= \begin{cases}
\min_{(u, v) \in E} \max (\en^*_{G, w'} (v) - w(u, v) + \en (v), 0) & \text{if $ u \in V_A $} \\
\max_{(u, v) \in E} \max (\en^*_{G, w'} (v) - w(u, v) + \en (v), 0) & \text{if $ u \in V_B $}
\end{cases} \, .
\end{align*}
Since $ \en (u) $ is a constant in the minimization and maximization terms, we get
\begin{align*}
\en^*_{G, w'} (u) &= \begin{cases}
\min_{(u, v) \in E} \max (\en^*_{G, w'} (v) - w(u, v) - \en (u) + \en (v), 0) & \text{if $ u \in V_A $} \\
\max_{(u, v) \in E} \max (\en^*_{G, w'} (v) - w(u, v) - \en (u) + \en (v), 0) & \text{if $ u \in V_B $}
\end{cases} \, .
\end{align*}
By the definition of $ w' $ this is equivalent to
\begin{align*}
\en^*_{G, w'} (u) &= \begin{cases}
\min_{(u, v) \in E} \max (\en^*_{G, w'} (v) - w'(u, v), 0) & \text{if $ u \in V_A $} \\
\max_{(u, v) \in E} \max (\en^*_{G, w'} (v) - w'(u, v), 0) & \text{if $ u \in V_B $}
\end{cases} \, .
\end{align*}
which is true by Lemma~\ref{lem:characterization_of_minimal_energy}.
\end{proof}

\iffalse
\begin{lemma}\label{lem:check_energy_correct_linear}
We can check whether an energy function $ e $ is the minimal energy of a graph $ (G, w) $ in linear time.
\end{lemma}

\begin{proof}
We use the characterization of the minimal energy provided by Lemma~\ref{lem:characterization_of_minimal_energy}.
We simply have to check whether all the conditions are fulfilled.
For every node we have to do work proportional to the number of outgoing edges.
Therefore the check can be done in $ O(m) $ time.
\end{proof}
\fi

\subsection{Full Proof of Lemma \ref{lem:approx to exact}}\label{sec:full proof of lem:approx to exact}

\begin{algorithm}
\caption{Computing minimal energy based on approximation}
\label{alg:exact_solution_by_approximation}

\DontPrintSemicolon
\LinesNotNumbered

\SetKwFunction{ValueIteration}{ValueIteration}
\SetKwFunction{Approximate}{Approximate}
\SetKwFunction{MinimalEnergy}{MinimalEnergy}
\SetKwFor{Function}{Procedure}{}{end}

\KwIn{A weighted graph $ (G, w) $, an upper bound $ M $ on the finite minimal energy of $ (G, w) $ and a lower bound~$ D $ on the penalty of $ (G, w) $}
\KwOut{The minimal energy of $ (G, w) $}

\BlankLine
\Function{$ \MinimalEnergy{G, w, M, D} $}{

\nl 	\uIf{$ D \leq M/(2n) $}{
\nl		\uIf{$ M \leq n $}{
			 \tcp{cf. Proposition~\ref{pro:value_iteration}}
\nl			\KwRet{\ValueIteration{$G, w, \{0, \ldots, n, \infty \}$}}\;
		}
\nl		\Else{
\nl			$ c \gets \lfloor \frac{M}{2} \rfloor $\;
\nl			$ \en \gets $ \Approximate{$G, w, M, c$} \tcp*[r]{cf. Theorem~\ref{thm:rounding_procedure_gives_desired_approximation}}
			\tcp{Now solve $ (G, w') $ with weights modified by energy $ \en $}
\nl			$ w'(u, v) \gets w(u, v) + \en(u) - \en(v) $ for every edge $ (u, v) \in G $\;
\nl			$ \en' \gets $ \MinimalEnergy{$G, w', c, D$}\;
\nl			$ \en'' (v) \gets  \en (v) + \en' (v) $ for every node $ v \in G $\;
\nl			\KwRet{$ \en'' $}
		}
	}
\nl	\Else{
\nl		$ c \gets nD $\;
\nl		$ \en \gets $ \Approximate{$G, w, M, c$} \tcp*[r]{cf. Theorem~\ref{thm:rounding_procedure_gives_desired_approximation}}
		\tcp{Now solve $ (G, w') $ with weights modified by energy $ \en $}
\nl		$ w'(u, v) \gets w(u, v) + \en(u) - \en(v) $ for every edge $ (u, v) \in G $\;
\nl		$ \en' \gets $ \MinimalEnergy{$G, w', c, D$}\;
\nl		$ \en'' (v) \gets  \en (v) + \en' (v) $ for every node $ v \in G $\;
\nl		\KwRet{$ \en'' $}
	}
}
\end{algorithm}

\iffalse
\setcounter{theoremtmp}{\value{theorem}}
\setcounter{theorem}{\value{LemmaThree}}
\begin{lemma}[Restated]
There is an algorithm that takes a graph $(G, w)$, a lower bound~$ D $ on the penalty $ P (G, w) $, and an upper bound $ M $ on the finite minimal energy of $ (G, w) $ as its input and computes the minimal energies of $ (G, w) $ in $ O(m n \log{D} + m \cdot \frac{M}{\lceil D \rceil})$ time. Specifically, for $ D \geq \frac{M}{2n} $ it runs in $ O(m n \log{(M/n)}) $ time.
\end{lemma}
\setcounter{theorem}{\value{theoremtmp}}
\fi

%\begin{proof}
Our algorithm is called \MinimalEnergy and is described in Algorithm~\ref{alg:exact_solution_by_approximation}.
We call the algorithm provided by Theorem~\ref{thm:rounding_procedure_gives_desired_approximation}, which computes an approximation of the minimal energy, \Approximate and we call the value iteration algorithm provided by Proposition~\ref{pro:value_iteration}, which computes the minimal energy exactly, \ValueIteration.

We first consider the case $ D \geq M/(2n) $.
As pointed out in the proof idea, the correctness of \MinimalEnergy in this case follows from Theorem~\ref{thm:rounding_procedure_gives_desired_approximation} and Lemma~\ref{lem:main_lemma_for_correctness}.
We therefore only argue about the running time.
If $ M \leq n $, we know that $ n $ is an upper bound on the finite minimal energy, and we can use the value iteration algorithm of Proposition~\ref{pro:value_iteration} with the admissible list $ \{0, \ldots, n, \infty \} $, as explained in Section~\ref{sec:value_iteration_algorithm}.
The running time in this case is $ O(m n) $.
The algorithm \Approximate runs in time $ O (m M n/c ) $ for the upper bound $ M $ on the finite minimal energy.
For $ c = \lfloor M/2 \rfloor $ the factor $ M $ cancels itself and therefore the running time of \Approximate is $ O (mn) $.
We recurse with the upper bound $ M' = c = \lfloor M/2 \rfloor $ on the finite minimal energy and the unchanged lower bound $ D $ on the penalty.
It is still the case that $ D \geq M'/(2n) $.
Thus, the running time of the procedure $ \MinimalEnergy $ is given by the following recurrence:
\begin{align*}
T \left(n, m, M \right) = \begin{cases}
O \left(m n \right) & \text{if $ M \leq n $} \\
T \left(n, m, \frac{M}{2} \right) + O \left(m n \right) & \text{otherwise}
\end{cases} \, .
\end{align*}
Since the initial value of $ M $ is halved with every iteration of the algorithm until $ M \leq n $, the algorithm runs for at most $ \log{M} - \log{n} = \log{(M/n)} $ many iterations.
Every iteration needs time $ O (mn) $ and therefore the total running time is $ O(m n \cdot \log{(M/n)}) $.

We now consider the case $ D < M/(2n) $ in which we perform one step of \Approximate to reduce $ M $ to $ M' $ such that $ D \geq M' / 2n $.
We first compute an approximation $ e $ of the minimal energy by calling \Approximate with the approximation error $ c = n D $.
Then we set $ w'(u, v) = w(u, v) + e(u) - e(v) $.
We can compute the approximation of the minimal energy in time $ O (m M /D) $.
After that we can recurse on $ (G, w') $ with the new upper bound $ M' = c = nD $ on the finite minimal energy to compute $ \en'(v) = \en_{G, w'} (v) $ for every node $ v $.
By Lemma~\ref{lem:main_lemma_for_correctness} the algorithm afterwards correctly returns the minimal energy $ \en''(v) = \en (v) + \en' (v) $ for every node $ v $.
The new upper bound $ M' $ fulfills the following inequality:
\begin{align*}
\frac{M'}{2n} = \frac{nD}{2n} = \frac{D}{2} < D \, .
\end{align*}
Since the penalty does not change, i.e., $ P(G, w) = P(G, w') $ by Lemma~\ref{lem:main_lemma_for_correctness}, our previous running time analysis of the case $ D \geq M'/(2n) $ now applies.
The remaining time needed to compute the minimal energy of $ (G, w') $ therefore is $ O(m n \cdot \log{(M'/n)}) = O(m n \cdot \log{D}) $.
Thus, the total running time in this case is $ O(m n \cdot \log{D} + m \cdot M/D) $.
This completes the proof of Lemma~\ref{lem:approx to exact}.
%\end{proof}

\section{Hardness on Complete Bipartite Graphs} 
\label{sec:discussion}

We show in this section that energy games on complete bipartite graphs are polynomial-time equivalent to the general case. This implies that energy games on graphs of bounded clique-width~\cite{CourcelleO00} and strongly ergodic\footnote{There are many notions of ergodicity \cite{Lebedev05,BorosEFGMM11}. Strong ergodicity is the strongest one as it implies other ergodicity conditions.} graphs~\cite{Lebedev05} are as hard as the general case.\footnote{We formally define the notion of \emph{clique-width} and the class of \emph{strongly ergodic} graphs in Section~\ref{subsec:clique width ergodicity}.}
Our result indicates that structural properties of the input graphs might not yield efficiently solvable subclasses.
This is in contrast to the fact that parity games (a natural subclass of energy and mean-payoff games) can be solved in polynomial time in these cases~\cite{Obdrzalek07,Lebedev05}.

Our main hardness result is for the \emph{decision} problem of energy games which will imply the hardness of the {\em value} problem as well as of mean-payoff games. The value problem is what we have discussed so far. 
The \emph{decision problem} of energy games for a graph $ (G, w) $ and a node $ s $ asks whether the minimal energy $e_{G, w}^*(s)$ is finite.
If $e_{G, w}^*(s)$ is finite, we say that Alice {\em wins at $s$}; otherwise, we say that Alice loses (or equivalently Bob wins).
The decision problem and the value problem of energy games are polynomial-time equivalent~\cite{Bouyer08}.\footnote{The reduction of Bouyer et al.~\cite{Bouyer08} adds nodes and edges such that, after every edge that is taken, Bob has the possibility to return to the starting node $ s $ by an edge of weight $ t $, for a finite $ t \geq 0 $.
In this way we have $ \en^*_{G, w} (s) \leq t $ if and only if Alice wins at $ s $.
It is now possible to find $ \en^*_{G, w} (s) $ by binary search because the maximum finite energy is limited to $ nW $.}

We show that the decision problem on strongly ergodic graphs or graphs of bounded clique-width is just as hard as the general decision problem on arbitrary graphs. For this purpose we will work with a special type of {\em complete bipartite} graphs which are strongly ergodic and have bounded clique-width (see Definition~\ref{def:complete bipartite}). We note the following fact, proved in Section~\ref{subsec:clique width ergodicity}. 
\begin{lemma}\label{lem:clique width and strongly ergodic}
Every complete bipartite graph has clique-width two and is strongly ergodic.
\end{lemma}
Our main result is a polynomial-time reduction from the decision problem on arbitrary graphs to the decision problem on complete bipartite graphs. 
\begin{theorem}\label{thm:equivalence_decision_complete_bipartite}\label{thm:complete bipartite}
The decision problem of energy games on complete bipartite graphs is polynomial-time equivalent to the decision problem of energy games on general graphs.
\end{theorem}
This shows that if we can solve the decision problem of energy games on very special graphs that have clique-width two and are strongly ergodic, then we can solve this problem on general graphs too. 

The relationship in Theorem~\ref{thm:equivalence_decision_complete_bipartite} also carries over to the value problem and to mean-payoff games. 
\begin{corollary}\label{cor:complete bipartite}
The value problem of energy games on complete bipartite graphs is polynomial-time equivalent to the value problem of energy games on general graphs.
Moreover, the mean-payoff game problem on complete bipartite graphs is polynomial-time equivalent to the mean-payoff game problem on general graphs.
\end{corollary}

The first statement of the above corollary follows from the fact that the value problem of energy games on general graphs can be reduced to the decision problem \cite{Bouyer08}, and the decision problem on complete bipartite graph is a special case of the value problem on complete bipartite graphs (because solving the value problem also answers the decision problem).
For the second statement observe that the decision problem of energy games and the decision problem of mean-payoff games are exactly the same problem~\cite{Bouyer08} because the minimal energy at a node $ v $ is finite if and only if the mean-payoff value at $ v $ is non-negative.
Therefore the statement follows from the fact that the value problem of mean-payoff games can be reduced to the decision problem of mean-payoff games~\cite{GurvichKK88}, and the decision problem of mean-payoff games on complete bipartite graphs is a special case of the value problem of mean-payoff games on complete bipartite graphs.

The rest of this section is devoted to proving Theorem~\ref{thm:complete bipartite}. We first give a proof idea in Section~\ref{subsec:proof idea complete bipartite}. In Section~\ref{subsec:clique width ergodicity}, we formally define the notion of complete bipartite graphs in the context of energy games and prove Lemma~\ref{lem:clique width and strongly ergodic}. In Section~\ref{subsec:reduce to win everywhere} and \ref{subsec:reduce to complete bipartite}, we show two parts of our reduction. In the first part (Section~\ref{subsec:reduce to win everywhere}), we reduce from the general decision problem of energy games to the problem where it is {\em promised} that one player wins everywhere, i.e., either the minimal energy function is finite at all nodes (Alice wins) or infinite at all nodes (Bob wins). Note that an input graph of this promised problem is still a general graph. In the second part (Section~\ref{subsec:reduce to complete bipartite}), we reduce from this win-everywhere problem on general graphs to the same problem on {\em complete bipartite} graphs.

\subsection{Proof Ideas of Theorem~\ref{thm:complete bipartite}}\label{subsec:proof idea complete bipartite}
The main idea of proving Theorem~\ref{thm:complete bipartite} is to add ``useless'' edges to the input graph to make the graph complete bipartite while the answer to the energy game problem remains the same. To illustrate this point, consider any input graph $(G, w)$ and a node $s$. For each node $u$ belonging to Alice, we add an edge $(u, v)$ with weight $-\infty$, for all nodes $v$ belonging to Bob. Let $(G', w')$ be the new graph.\footnote{Readers that are familiar with parity games might wonder why the same idea does not work for parity games. In parity games every node has a priority. This corresponds to the case where all outgoing edges of a node have the same weight. Under this restriction we are not allowed to add edges of weight $-\infty$ wherever we want to.}

Observe that if Alice wins in $(G, w)$, i.e. $\en^*_{G, w}(s)<\infty$, then she can still play the same strategy in $(G', w')$ so that she wins in $(G', w')$, i.e. $\en^*_{G', w'}(s)<\infty$. 
On the other hand, if Alice loses in $(G, w)$, i.e. $\en^*_{G, w}(s)=\infty$, then the only way she can win in $(G', w')$ is to use some edges that are not in $(G, w)$. These edges, however, have weight $-\infty$. So, the minimum energy that Alice requires remains $\infty$ even when she use the new edges. Thus, Alice also loses in $(G', w')$, i.e.  $\en^*_{G', w'}(s)=\infty$. 

The actual proof of Theorem~\ref{thm:complete bipartite} is based on this idea but needs a bit more work. This is because we cannot add an edge of weight $-\infty$. We instead add an edge of weight $-X$, for large enough $X$.
But this does not solve the whole problem since, when $\en^*_{G, w}(s)=\infty$, Alice can use this edge to ``escape'' to some node $v$ such that $\en^*_{G, w}(s)<\infty$. This will make $\en^*_{G', w'}(s)<\infty$. To get around this, we first reduce the problem on $(G, W)$ to another graph $(G'', w'')$ where Alice either wins everywhere or loses everywhere. This makes the escaping impossible. We do this in Section~\ref{subsec:reduce to win everywhere}. After we have reduced to the case where one player wins everywhere, we can add edges as above. We also have to add edges from Bob's nodes to Alice's nodes. We assign to these edges a large \emph{positive} weight.

\subsection{Properties of Complete Bipartite Graphs}\label{subsec:clique width ergodicity}

In the following we define what we mean by the class of complete bipartite graphs and show that these graphs are strongly ergodic and have clique-width two.
Later, we will show that we can reduce energy games on arbitrary graphs to energy games on complete bipartite graphs.

\begin{definition}\label{def:complete bipartite}
A \emph{complete bipartite graph} is a graph $ G = (V, E) $ fulfilling the following two conditions:
\begin{itemize}
\item \emph{(bipartite)} There is no edge $ (u, v) $ from a node $ u \in V_A $ of Alice to a node $ v \in V_A $ of Alice and there is no edge $ (u, v) $ from a node $ u \in V_B $ of Bob to a node $ v \in V_B $ of Bob.
\item \emph{(complete)} For every node $ u \in V_A $ of Alice and every node $ v \in V_B $ there is an edge $ (u, v) \in E $ and for every node $ u \in V_B $ of Bob and every node $ v \in V_A $ of Bob there is an edge $ (u, v) $.
\end{itemize}
\end{definition}
Note that the number of nodes of Alice and Bob is not required to be equal to fit this definition.
We claim that every complete bipartite graph has clique-width two and is strongly ergodic.

The notion of clique-width was introduced by Courcelle and Olariu~\cite{CourcelleO00}.
We state the definition of clique-width using different notation.
\begin{definition}
The \emph{clique-width} of a graph is the minimum number of labels needed to construct $ G $ by means of the following four operations.
\begin{enumerate}
\item Creation of a new node with label $ i $
\item Disjoint union of two labeled graphs
\item Adding an edge $ (u, v) $ for every vertex $ u $ with label $ i $ and every vertex $ v $ with label $ j $
\item Renaming label $ i $ to label $ j $
\end{enumerate}
\end{definition}
It is easy to see that every complete bipartite graph has clique-width~2.
Note that $ 2 $ is the smallest clique-width possible for a graph with more than one node.
Furthermore, every graph that has bounded tree-width also has bounded clique-width~\cite{CourcelleO00}.
The concept of tree-width is applied to directed graphs by viewing every edge as an undirected edge.
Remember that parity games, which can be reduced to energy games in polynomial-time, can be solved in polynomial time on graphs of bounded clique-width~\cite{Obdrzalek07}.

We now show that every complete bipartite graph is strongly ergodic.
\begin{definition}[Ergodicity]\label{def:ergodic}\footnote{We use Lebedev's definitions~\cite{Lebedev05}.}
An \emph{ergodic partition} is a pair $ (S_A, S_B) $ of sets of nodes such that $ S_A $ and $ S_B $ are a partition of the nodes satisfying the following conditions:
\begin{enumerate}
\item For every node $ u $ in $ S_A \cap V_A $ there is a node $ v \in S_A $ such that $ (u, v) \in E $; i.e., Alice can always keep the car inside $S_A$ if she wants to.
\item There is no edge $ (u, v) $ such that $ u \in S_A \cap V_B $ and $ v \in S_B $; i.e., Bob cannot move the car out of $S_A$.
\item For every node $ u $ in $ S_B \cap V_B $ there is a node $ v \in S_B $ such that $ (u, v) \in E $; i.e., Bob can always keep the car inside $S_B$ if he wants to.
\item There is no edge $ (u, v) $ such that $ u \in S_B \cap V_A $ and $ v \in S_A $; i.e., Alice cannot move the car out of $S_B$.
\end{enumerate}

A graph is \emph{ergodic} if it has no non-trivial ergodic partition (a partition $ (S_A, S_B) $ is trivial if $ S_A = \emptyset $ or $ S_B = \emptyset $).
A graph is \emph{strongly ergodic} if every induced subgraph such that every node has out-degree at least $ 1 $ is ergodic.
\end{definition}
\begin{lemma}
Every complete bipartite graph is strongly ergodic.
\end{lemma}
\begin{proof}
Note that every induced subgraph of a complete bipartite graph is also a complete bipartite graph.
Therefore it is sufficient to show that every complete bipartite graph is ergodic.

Suppose that there is a complete bipartite graph $G$ that is not ergodic.
Then $ G $ has a non-trivial ergodic partition $ (S_A, S_B) $.
We consider three cases where each one leads to a contradiction:
\begin{itemize}
\item $ S_A $ contains a node $ u $ of Bob, and $ S_B $ contains a node $ v $ of Alice:
Since we have a complete bipartite graph there is an edge $ (u, v) $.
This means that Bob has an edge leaving $ S_A $ which contradicts the second condition in Definition~\ref{def:ergodic}.
\item $ S_B $ contains no node of Alice ($S_A$ might or might not contain a node of Bob):
Then $ S_B $ only contains nodes of Bob.
Since the graph is bipartite all nodes of $ S_B $ only have edges that leave $ S_B $.
Since $ S_B $ is nonempty, there is a node of Bob in $ S_B $ that has no edge that stays in $ S_B $ which contradicts the third condition in Definition~\ref{def:ergodic}.
\item $ S_A $ contains no node of Bob ($S_B$ might or might not contain a node of Alice): symmetric to previous case.
\end{itemize}
Since $ S_A \neq \emptyset $ and $ S_B \neq \emptyset $ we have considered all cases.
\end{proof}
We remark that every graph that is \emph{strongly ergodic} is also \emph{structurally ergodic} in the sense of Boros et al.~\cite{BorosEFGMM11}.
Thus, complete bipartite graphs are also structurally ergodic.

\subsection{Reduction to Graphs Where One Player Wins Everywhere}\label{subsec:reduce to win everywhere}

In the following, we give the first reduction.
We show that energy games on arbitrary weighted graphs can---in polynomial time---be reduced to energy games on weighted graphs in which one player wins at every node.

\begin{lemma}
For energy games, the following variants of the decision problem are polynomial-time equivalent:
\begin{itemize}
\item Decision problem on arbitrary weighted graphs.
\item Decision problem on weighted graphs in which one player wins everywhere.
\end{itemize}
\end{lemma}

Clearly, graphs in which one player wins everywhere are included in the class of all graphs.
The reduction from arbitrary graphs to graphs in which one player wins everywhere goes as follows.
We are given a graph $ (G, w) $ and want to solve the decision problem, i.e., we want to figure out which player wins at a node $ s $ in $ (G, w) $.
We construct a graph $ (G', w') $ as follows.
All nodes of $ G $ also appear in $ G' $ and belong to the same player as in $ G $.
We replace every edge $ (x, y) $ of $ G $ (see Fig.~\ref{fig:reduction_one_player_wins_everywhere}) by the following construction:
We add a node $ u $ of Alice and node $ v $ of Bob and add the edges $ (x, u) $, $ (u, v) $, $ (v, y) $, $ (u, s) $, and $ (v, s) $ with the weights $ w' (x, u) = w (x, y) $, $ w' (u, v) = w' (v, y) = 0 $, $ w' (u, s) = -n W $, and $ w' (v, s) = nW $.

\begin{figure}[htbp!]
\centering
\begin{tikzpicture}
\SetGraphUnit{2.5}
\GraphInit[vstyle=Normal]
\tikzset{LabelStyle/.style= {draw=none,inner sep=3pt,outer sep=0pt,fill=lightgray!50}}
\tikzset{EdgeStyle/.style = {->}}

\tikzset{VertexStyle/.append style = {shape = diamond}}
\Vertex[L=$x$]{x}
\EA[L=$y$](x){y}

\Edge[label=$w(x{,}y)$](x)(y)
\end{tikzpicture}

$ \Downarrow $

\begin{tikzpicture}
\SetGraphUnit{2.5}
\GraphInit[vstyle=Normal]
\tikzset{LabelStyle/.style= {draw=none,inner sep=3pt,outer sep=0pt,fill=lightgray!50}}
\tikzset{EdgeStyle/.style = {->}}

\tikzset{VertexStyle/.append style = {shape = diamond}}
\Vertex[L=$x$]{x}
\tikzset{VertexStyle/.append style = {shape = circle}}
\EA[L=$u$](x){u}
\tikzset{VertexStyle/.append style = {shape = diamond}}
\SOEA[L=$s$](u){s}
\tikzset{VertexStyle/.append style = {shape = rectangle}}
\NOEA[L=$v$](s){v}
\tikzset{VertexStyle/.append style = {shape = diamond}}
\EA[L=$y$](v){y}

\Edge[label=$w(x{,}y)$](x)(u)
\Edge[label=$0$](u)(v)
\Edge[label=$0$](v)(y)
\Edge[label=$-nW$](u)(s)
\Edge[label=$nW$](v)(s)
\end{tikzpicture}
\caption{This picture shows the reduction from the general decision problem to the decision problem in which one of the players wins everywhere.
The round nodes belong to Alice and the rectangular nodes belong to Bob.
The diamond-shaped nodes are unspecified and could belong to any of the two players.}
\label{fig:reduction_one_player_wins_everywhere}
\end{figure}
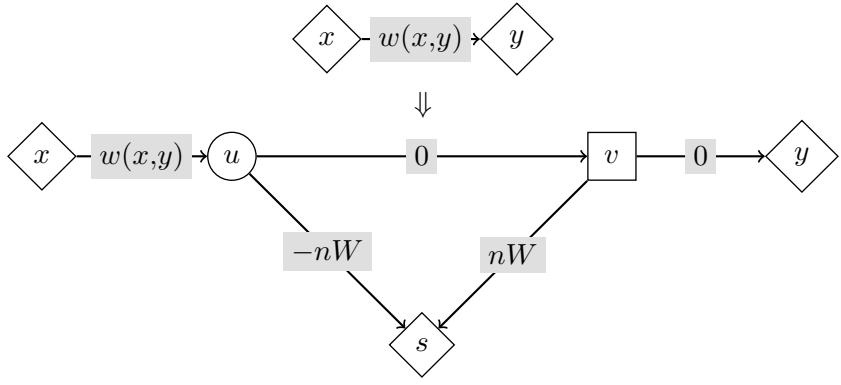

\begin{lemma}\label{lem:same_winning_at_starting_node}
Alice wins at $ s $ in $ (G, w) $ if and only if Alice wins at $ s $ in $ (G', w') $.
\end{lemma}
\begin{proof}
We first prove the following claim: If Alice wins at $ s $ in $ (G, w) $, then Alice also wins at $ s $ in $ (G', w') $.
Alice simply has to play the winning strategy $ \sigma^* $ for $ s $ in $ (G, w) $.\footnote{To be precise: Alice has to play $ \sigma^* $ for nodes already present in $ (G, w) $ and for the other nodes the edge that does \emph{not} go back to $ s $ has to be chosen.}
If Bob never plays a new edge that goes back to $ s $, his strategy was also available in $ (G, w) $ and then Alice wins because $ \sigma^* $ is a winning strategy in $ G $. 
As soon as Bob plays one of the new edges, a cycle is formed.
The cycle $ C $ consists of a simple path $ P $ from $ s $ to some node $ v $ and then an edge from $ v $ to $ s $.
Since the path $ P $ in $ (G', w') $ does not contain an edge going to $ s $, it corresponds to some path in $ (G, w) $ of the same weight.
As a simple path in $ (G, w) $ contains at most $ n-1 $ edges each of weight at least $ -W $, the weight of $ P $ is at least $ -(n-1) W $.
Since the edge from $ x $ to $ s $ has weight $ nW $, the cycle $ C $ has positive weight.
Therefore $ \sigma^* $ is also a winning strategy in $ (G', w') $.

A symmetric argument can be used to prove the following claim: If Bob wins at $ s $ in $ (G, w) $, then Bob also wins at $ s $ in $ (G', w') $.
Now the lemma follows from determinacy: Alice does not win if and only if Bob wins.
\end{proof}
\begin{lemma}
One of the players wins everywhere in $ (G', w') $.
\end{lemma}
\begin{proof}
We show that the player that wins at $ s $ in $ (G, w) $ is the one that wins everywhere in $ G' $.
We assume that Alice wins at $ s $ in $ (G, w) $.
(For Bob the argument is symmetric.)
By Lemma~\ref{lem:same_winning_at_starting_node} it follows that Alice wins at $ s $ in $ (G', w') $ by playing some strategy $ \sigma $.
We define a strategy $ \sigma' $ for every node $ v $ of Alice as follows:
If the edge $ (v, s) $ does not exist, we set $ \sigma'(v) = \sigma(v) $.
If the edge $ (v, s) $ does exist we distinguish two cases.
If Alice wins at $ v $ in $ (G', w') $ by playing according to $ \sigma $, then $ \sigma'(v) = \sigma(v) $.
Otherwise, Alice takes the new edge that goes to $ s $, i.e., $ \sigma(v) = s $.
In other words, $ \sigma' $ is defined as follows for every node $ v $ of Alice:
\begin{align*}
\sigma'(v) =
\begin{cases}
s & \text{if edge $ (v, s) $ exists in $ G' $ and Alice loses at $ v $ in $ (G', w') $ by playing $ \sigma $} \\
\sigma (v) & \text{otherwise}
\end{cases}
\end{align*}

We now show that with the strategy $ \sigma' $ Alice wins against any strategy $ \tau $ of Bob.
Let $ P $ be the (unique) infinite path in $ (G'(\sigma, \tau), w') $ starting at $ s $.\footnote{Because of its special structure such a path $ P $ is also known as a ``lasso'' in the literature.}
Since $ \sigma $ is a winning strategy of Alice starting from $ s $ in $ (G', w') $, Alice wins for every node on $ P $ in $ (G', w') $ by playing according to $ \sigma $.
By the above definition of $ \sigma' $ we have $ \sigma'(v) = \sigma (v) $ for every node $ v $ on $ P $.
This means that the infinite path in $ (G'(\sigma', \tau), w') $ starting at $ s $ is exactly $ P $ and contains a non-negative cycle.

We now show that in fact for every node $ u $, the infinite path $ P' $ in $ (G'(\sigma', \tau), w') $ starting at $ u $ contains a non-negative cycle.
If $ P' $ contains $ s $, then $ P' $ ends in $ P $.
As argued above, $ P $ contains a non-negative cycle and therefore also $ P' $ contains a non-negative cycle.
Consider now the case that $ P' $ does not contain $ s $ which implies that $ \sigma'(v) = \sigma (v) $ for every node $ v $ of Alice on $ P' $ (because otherwise $ P' $ would contain $ s $).
Therefore $ P' $ is equal to the infinite path in $ (G'(\sigma, \tau), w') $ starting at $ u $.
By the way we constructed $ G' $, $ P' $ must contain at least one node $ v $ of Alice that has an edge $ (v, s) $ to $ s $.
Since $ \sigma' (v) = \sigma (v) \neq s $ it follows by the way we defined $ \sigma' $ that Alice wins at $ v $ in $ (G', w') $ by playing according to $ \sigma $.
Therefore $ P'$ contains a non-negative cycle as desired.
Since $ \tau $ was an arbitrary strategy of Bob, we know that Alice wins everywhere in $ (G', w') $ with the strategy $ \sigma' $.
\end{proof}

\subsection{Reduction to Complete Bipartite Graphs}\label{subsec:reduce to complete bipartite}

We now give our second reduction.
We show how to reduce the decision problem on graphs in which one player wins everywhere to the decision problem on complete bipartite graphs, as in the following lemma. 

\begin{lemma}
For energy games, the following variants of the decision problem are polynomial-time equivalent.
\begin{enumerate}[(1)]
\item Decision problem on graphs in which one player wins everywhere. 
\item Decision problem on weighted complete bipartite graphs. 
\end{enumerate}
\end{lemma}

Note that the reduction from (2) to (1) is trivial. This is because complete bipartite graphs are strongly ergodic, and in strongly ergodic graphs one player wins everywhere (because otherwise the sets of winning nodes of Alice and Bob, respectively, would immediately give a non-trivial ergodic partition).% Thus, a weighted complete bipartite graph is already an instance where one player wins everywhere. 

The rest of this subsection is devoted to showing the reduction from (1) to (2). This reduction has two parts.  
We first reduce from (1) to bipartite graphs, which can be done very easily, and from there we reduce to complete bipartite graphs.

\subsubsection{Part 1: Reduction to Bipartite Graphs}
We are given a graph $ (G, w) $ in which one of the players wins everywhere.
We want make the graph bipartite, i.e., there should neither be an edge $ (u, v) $ such that $ u \in V_A $ and $ v \in V_A $ nor should there be and edge $ (u, v) $ such that $ u \in V_B $ and $ v \in V_B $.
We modify $ (G, w) $ as follows:
\begin{itemize}
\item We replace every edge $ (u, v) \in E $ such that $ u, v \in V_A $ by two edges $ (u, u') $ and $ (u', v)$ where $ u' $ is a new node of Bob and the weights of the new edges are $ w_0 (u, u') = w(u, v) $ and $ w_0 (u', v) = 0 $.
\item We replace every edge $ (u, v) \in E $ such that $ u, v \in V_B $ by two edges $ (u, u') $ and $ (u', v)$ where $ u' $ is a new node of Alice and the weights of the new edges are $ w_0 (u, u') = w(u, v) $ and $ w_0 (u', v) = 0 $.
\end{itemize}
We call the resulting graph $ (G_0, w_0) $.
Observe that $ \en^*_{G, w} (v) = \en^*_{G_0, w_0} (v) $ for every node $ v $ of $ G $.
Therefore the player that wins everywhere in $ (G, w) $ also wins everywhere in $ (G_0, w_0) $.
The same reduction has recently been considered for mean-payoff games by Boros et al.~\cite{BorosEGM13}.

\subsubsection{Part 2: Reduction to Complete Bipartite Graphs}
We are given a bipartite graph $ (G, w) $ in which one player wins everywhere.
The reduction to complete bipartite graphs has two steps:
\begin{enumerate}
\item Modification of $ (G, w) $: For every pair $ (u, v) $ of nodes such that $ u \in V_A $ and $ v \in V_B $, if the edge $ (u, v) $ is not contained in $ G $, we add it with weight $ w_1 (u,v) = - nW $.
We call the resulting graph $ (G_1, w_1) $.
\item Modification of $ (G_1, w_1) $: For every pair $ (u, v) $ of nodes such that $ u \in V_B $ and $ v \in V_A $, if the edge $ (u, v) $ is not contained in $ G $, we add it with weight $ w_2 (u,v) = n^2 W $.
We call the resulting graph $ (G_2, w_2) $.
\end{enumerate}
Clearly $ (G_2, w_2) $ is a complete bipartite graph.

\begin{lemma}
In $ (G, w) $ and $ (G_2, w_2) $ the same player wins everywhere.
\end{lemma}

\begin{proof}
The following claims follow easily:
\begin{itemize}
\item If Alice wins everywhere in $ (G, w) $, then Alice also wins everywhere in $ (G_1, w_1) $. (Alice simply has to play the same strategy as in $ (G, w) $, Bob does not have more strategies than in $ (G, w) $.)
\item If Bob wins everywhere in $ (G_1, w_1) $, then Bob also wins everywhere in $ (G_2, w_2) $. (Bob simply has to play the same strategy as in $ (G_1, w_1) $, Alice does not have more strategies than in $ (G_1, w_1) $.)
\end{itemize}

Now we show the following: If Bob wins everywhere in $ (G, w) $, then Bob also wins everywhere in $ (G_1, w_1) $.
Let $ \tau^* $ be a winning strategy of Bob in $ (G, w) $.
We argue that $ \tau^* $ is also a winning strategy of Bob in $ (G_1, w_1) $.
Let $ \sigma_1 $ be an arbitrary strategy of Alice in $ (G_1, w_1) $.
Let $ C $ be a cycle in $ G_1 (\sigma_1, \tau^*) $.
If all edges of $ C $ already occur in $ (G, w) $, we know that $ C $ is a cycle of negative weight in $ (G_1, w_1) $ because $ \tau^* $ is a winning strategy of Bob in $ G $.
If there is an edge in $ C $ that did not already occur in $ (G, w) $, then this edge has weight $ -nW $.
Since the largest positive weight in $ (G_1, w_1) $ is $ W $, and $ C $ consists of at most $ n $ edges, we know that $ C $ is a cycle of negative weight.
Thus, every cycle in $ G (\sigma_1, \tau^*) $ has negative weight.
Since $ \sigma_1 $ was an arbitrary strategy of Alice in $ (G_1, w_1) $, we conclude that $ \tau^* $ is a winning strategy of Bob in $ (G_1, w_1) $ which he can play to win everywhere.

A symmetric argument can be used to prove the following:
If Alice wins everywhere in $ (G_1, w_1) $, then Alice also wins everywhere in $ (G_2, w_2) $.
The only difference to before is that the minimal negative edge weight in $ (G_1, w_1) $ is $ - nW $ which is the reason why we have to set the weight of the new edges of Bob to $ n^2 W $.
Since either Alice wins everywhere in $ (G, w) $ or Bob wins everywhere in $ (G, w) $ it follows by our claims that the same player wins everywhere in $ (G, w) $ and $ (G_2, w_2) $.
\end{proof}

\section{Conclusion}

In this paper we answer the question whether the energy game problem can be solved efficiently under certain restrictions.
We give both negative and positive answers to this question.
On the negative side, we show that usual {\em graph structure restrictions}, namely clique-width and strong ergodicity, do not make the problem easier.
This is in contrast to the situation of the parity game problem (a special case of the energy game problem), which can be solved in polynomial time under such restrictions.
Thus, our result provides evidence that energy games might really be harder to solve than parity games.

On the positive side, we identify two {\em weight structure restrictions} that allow us to solve the energy game problem efficiently: fixed-window and large penalty restrictions. 
We also provide an algorithm for solving the energy game problem with additive error and show how to use this algorithm to solve the energy game problem exactly. 

Many problems remain open for solving energy games and related problems. The most fundamental one is, of course, settling the complexity status of these problems. On the one hand, current algorithmic techniques seem to be insufficient to show that these problems can be solved in polynomial time. On the other hand, it is unlikely that these problem are hard for any complexity classes currently known. 
It is interesting to investigate how weight structures can help in attacking these problems. For example, it might be possible to transform a graph $(G, w)$ to another graph $(G', w')$ whose penalty is large while the solution to the energy game problem remains the same. While this might be true for {\em any} graph $(G, w)$, we believe that it is already interesting to show this for some natural class of graphs, e.g. bounded tree-width graphs and graphs from the special case of parity games.

\printbibliography[heading=bibintoc] % Make bibliography show up in table of contents

\end{document}